\documentclass[orivec,conference]{llncs}

\newcommand{\pnt}[1]{{\mbox{$\vec{#1}$}}}
\newcommand{\ppnt}[2]{{\mbox{$\vec{#1}_{#2}$}}}

\newcommand{\V}[1]{\mbox{$\mathit{Vars}(#1)$}}

\newcommand{\s}[1]{\mbox{$\{#1\}$}}

\newcommand{\nGz}[2]{$G_{non-\{z\}}$}

\newcommand{\prr}[1]{\mi{Prev}(\boldsymbol{q})}

\newcommand{\mi}[1]{\mathit{#1}}
\newcommand{\ti}[1]{\textit{#1}}
\newcommand{\tb}[1]{\textbf{#1}}

\newcommand{\ttt}{\>\>\>}

\newcommand{\Tt}{\>\>}

\newcommand{\sub}[2]{\mbox{$\vec{#1}_\mi{#2}$}}

\newcommand{\prob}[2]{\mbox{$\exists{#1} [#2]$}}

\newcommand{\Comment}[1]{}

\newcommand{\sas}{\mbox{$\mi{SemStr}$}\xspace}
\newcommand{\cube}[1]{\mbox{$\mi{Cube}(\pnt{#1})$}}
\newcommand{\Cube}[2]{\mbox{$\mi{Cube}(\ppnt{#1}{#2})$}}
\newcommand{\nbhd}[2]{\mbox{$\mi{Nbhd}(\pnt{#1},#2)$}}
\newcommand{\nnbhd}[3]{\mbox{$\mi{Nbhd}(\ppnt{#1}{#2},#3)$}}
\newcommand{\Nbhd}[3]{\mbox{$\mi{Nbhd}(\sub{#1}{init},\pnt{#2},#3)$}}
\newcommand{\NNbhd}[4]{\mbox{$\mi{Nbhd}(\sub{#1}{init},\ppnt{#2}{#3},#4)$}}
\newcommand{\Nbh}[3]{\mbox{$\mi{Nbhd}(\pnt{#1},\pnt{#2},#3)$}}
\newcommand{\ac}[1]{\mbox{$\Phi(#1)$}}
\newcommand{\Fi}{\mbox{$\Phi$}\xspace}

\usepackage[cmex10]{amsmath}
\usepackage{amsfonts}
\usepackage{appendix}
\usepackage{wrapfig}
\usepackage{graphicx}
\usepackage{enumerate}
\usepackage{multirow}
\usepackage{xspace}
\usepackage{cancel}
\begin{document}

\title{Generation Of Complete Test Sets}


\author{Eugene Goldberg}
\institute{\email{eu.goldberg@gmail.com}}

\maketitle

\begin{abstract}
  We use testing to check if a combinational circuit $N$ always
  evaluates to 0 (written as $N \equiv 0$).  The usual point of view
  is that to prove $N \equiv 0$ one has to check the value of $N$ for
  all $2^{|X|}$ input assignments where $X$ is the set of input
  variables of $N$. We use the notion of a Stable Set of Assignments
  (SSA) to show that one can build a \ti{complete} test set (i.e. a
  test set proving $N \equiv 0$) that consists of less than $2^{|X|}$
  tests. Given an unsatisfiable CNF formula $H(W)$, an SSA of $H$ is a
  set of assignments to $W$ proving unsatisfiability of $H$. A trivial
  SSA is the set of all $2^{|W|}$ assignments to $W$. Importantly,
  real-life formulas can have SSAs that are much smaller than
  $2^{|W|}$.  Generating a complete test set for $N$ using only the
  machinery of SSAs is inefficient. We describe a much faster
  algorithm that combines computation of SSAs with resolution
  derivation and produces a complete test set for a ``projection'' of
  $N$ on a subset of variables of $N$.  We give experimental results
  and describe potential applications of this algorithm.
 
\end{abstract}


\section{Introduction}
  Testing is an important part of verification flows. For that reason,
  any progress in understanding testing and improving its quality is
  of great importance. In this paper, we consider the following
  problem. Given a single-output combinational circuit $N$, find a set
  of input assignments (tests) proving that $N$ evaluates to 0 for
  every test (written as $N \equiv 0$) or find a
  counterexample\footnote{Circuit $N$ usually describes some property of a multi-circuit $M$,
the latter being the real object of verification. For instance, $N$
may specify a requirement that $M$ never outputs some combinations of
values.
}. We will call a set of
  input assignments proving $N \equiv 0$ a \ti{complete test set}
  (\ti{CTS})\footnote{Term CTS is sometimes used to say that a test set is complete in terms
of a \ti{coverage metric} i.e. that every event considered by this
metric is tested. Our application of term CTS is obviously quite
different.
}.  We will call a CTS
  \ti{trivial} if it consists of all possible tests.  Typically, one
  assumes that proving $N \equiv 0$ involves derivation of a trivial
  CTS, which is infeasible in practice.  Thus, testing is used only
  for finding an input assignment refuting $N \equiv 0$. In this
  paper, we present an approach for building a non-trivial CTS that
  consists only of a subset of all possible tests.

  Let $N(X,Y,z)$ be a single-output combinational circuit where $X$
  and $Y$ are sets of variables specifying input and internal
  variables of $N$ respectively. Variable $z$ specifies the output of
  $N$. Let $F_N(X,Y,z)$ be a formula defining the functionality of $N$
  (see Section~\ref{sec:cts}). We will denote the set of variables of
  circuit $N$ (respectively formula $H$) as \V{N} (respectively
  \V{H}). Every assignment\footnote{By an assignment to a set of variables $V$, we mean a \ti{full}
assignment where every variable of $V$ is assigned a value.
} to \V{F_N}
  satisfying $F_N$ corresponds to a consistent
  assignment\footnote{An assignment to a gate $G$ of $N$ is called consistent if the value
assigned to the output variable of $G$ is implied by values assigned
to its input variables. An assignment to variables of $N$ is called
consistent if it is consistent for every gate of $N$.
} to \V{N} and vice versa. Then
  the problem of proving $N \equiv 0$ reduces to showing that formula
  $F_N \wedge z$ is unsatisfiable.  From now on, we assume that all
  formulas mentioned in this paper are \ti{propositional}. Besides, we
  will assume that every formula is represented in CNF i.e. as a
  conjunction of disjunctions of literals.  We will also refer to a
  disjunction of literals as a \ti{clause}.

  Our approach is based on the notion of a Stable Set of Assignments
  (SSA) introduced in~\cite{ssp}.  Given formula $H(W)$, an SSA of $H$
  is a set $P$ of assignments to variables of $W$ that have two
  properties.  First, every assignment of $P$ falsifies $H$. Second,
  $P$ is a transitive closure of some neighborhood relation between
  assignments (see Section~\ref{sec:ssa}). The fact that $H$ has an
  SSA means that the former is unsatisfiable. Otherwise, an assignment
  satisfying $H$ is generated when building its SSA. If $H$ is
  unsatisfiable, the set of all $2^{|W|}$ assignments is always an SSA
  of $H$ . We will refer to it as \ti{trivial}. Importantly, a
  real-life formula $H$ can have a lot of SSAs whose size is much less
  than $2^{|W|}$. We will refer to them as \ti{non-trivial}.  As we
  show in Section~\ref{sec:ssa}, the fact that $P$ is an SSA of $H$ is
  a \ti{structural} property of the latter. That is this property
  cannot be expressed in terms of the truth table of $H$ (as opposed
  to a \ti{semantic} property of $H$). For that reason, if $P$ is an
  SSA for $H$, it may not be an SSA for some other formula $H'$ that
  is logically equivalent to $H$.

  We show that a CTS for $N$ can be easily extracted from an SSA of
  formula $F_N \wedge z$. This makes a non-trivial CTS a structural
  property of circuit $N$ that cannot be expressed in terms of its
  truth table.  Unfortunately, building an SSA even for a formula of
  small size is inefficient.  To address this problem, we present a
  procedure that constructs a simpler formula $H(V)$ where $V
  \subseteq \V{F_N \wedge z}$ for which an SSA is generated.  Formula
  $H$ is implied by $F_N \wedge z$. Thus, the unsatisfiability of $H$
  proved by construction of its SSA implies that $F_N \wedge z$ is
  unsatisfiable too and $N \equiv 0$.  A test set extracted from an
  SSA of $H$ can be viewed as a CTS for a ``projection'' of $N$ on
  variables of $V$.

  We will refer to the procedure for building formula $H$ above as
  \sas (``\ti{Sem}antics and \ti{Str}ucture''). The name is due to the
  fact that \sas combines semantic and structural derivations.  \sas
  can be applied to an arbitrary CNF formula $G(V,W)$. If $G$ is
  unsatisfiable, \sas returns a formula $H(V)$ implied by $G$ and its
  SSA. Otherwise, it produces an assignment to $V \cup W$ satisfying
  $G$. The semantic part of \sas is to derive $H$. Its structural part
  consists of proving that $H$ is unsatisfiable by constructing an
  SSA. Formula $H$ produced when $G$ is unsatisfiable is logically
  equivalent to \prob{W}{G}. Thus, \sas can be viewed as a quantifier
  elimination algorithm for unsatisfiable formulas. On the other hand,
  \sas can be applied to check satisfiability of a CNF formula, which
  makes it a SAT-algorithm.

  The notion of non-trivial CTSs helps better understand testing.  The
  latter is usually considered as an incomplete version of a semantic
  derivation. This point of view explains why testing is efficient
  (because it is incomplete) but does not explain why it is effective
  (only a minuscule part of the truth table is sampled). Since a
  non-trivial CTS for $N$ is its structural property, it is more
  appropriate to consider testing as a version of a \ti{structural}
  derivation (possibly incomplete). This point of view explains not
  only efficiency of testing but provides a better explanation for its
  effectiveness: by using circuit-specific tests one can cover a
  significant part of a non-trivial CTS.

  The contribution of this paper is threefold. First, we use the
  machinery of SSAs to introduce the notion of non-trivial CTSs
  (Section~\ref{sec:cts}). Second, we present \sas, a SAT-algorithm
  that combines structural and semantic derivations
  (Section~\ref{sec:algor}).  We show that this algorithm can be used
  for computing a CTS for a projection of a circuit.  We also discuss
  some applications of \sas (Sections~\ref{sec:appl_test}
  and~\ref{sec:appl_sat}).  Third, we give experimental results
  showing the effectiveness of tests produced by \sas
  (Section~\ref{sec:exper}). In particular, we describe a procedure
  for ``piecewise'' construction of test sets that can be potentially
  applied to very large circuits.

\section{Stable Set Of Assignments}
\label{sec:ssa}
\subsection{Some definitions}
Let \pnt{p} be an assignment to a set of variables $V$. Let \pnt{p}
falsify a clause $C$.  Denote by {\boldmath \nbhd{p}{C}} the set of
assignments to $V$ satisfying $C$ that are at Hamming distance 1 from
\pnt{p}. (Here \ti{Nbhd} stands for ``Neighborhood''). Thus, the
number of assignments in \nbhd{p}{C} is equal to that of literals in
$C$. Let \pnt{q} be another assignment to $V$ (that may be equal to
\pnt{p}). Denote by {\boldmath \Nbh{q}{p}{C}} the subset of
\nbhd{p}{C} consisting only of assignments that are farther away from
\pnt{q} than \pnt{p} (in terms of the Hamming distance).

\begin{example}
  Let $V=\s{v_1,v_2,v_3,v_4}$ and \pnt{p}=0110. We assume that the
  values are listed in \pnt{p} in the order the corresponding
  variables are numbered i.e. \mbox{$v_1=0$}, $v_2=1,v_3=1,v_4=0$. Let $C= v_1
  \vee \overline{v_3}$. (Note that \pnt{p} falsifies $C$.) Then
  \nbhd{p}{C}=\s{\ppnt{p}{1},\ppnt{p}{2}} where \ppnt{p}{1} = 1110 and
  \ppnt{p}{2}=0100. Let \pnt{q} = 0000. Note that \ppnt{p}{2} is
  actually closer to \pnt{q} than \pnt{p}. So
  \Nbh{q}{p}{C}=\s{\ppnt{p}{1}}.
\end{example}
\begin{definition}
  \label{def:ac_fun}
  Let $H$ be a formula\footnote{In this paper, we use the set of clauses \s{C_1,\dots,C_k} as an
alternative representation of a CNF formula $C_1 \wedge \dots \wedge
C_k$.
\vspace{-10pt}

} specified by a set
  of clauses \s{C_1,\dots,C_k}.  Let $P$ =
  \s{\ppnt{p}{1},\dots,\ppnt{p}{m}} be a set of assignments to \V{H}
  such that every $\ppnt{p}{i} \in P$ falsifies $H$.  Let \Fi denote a
  mapping $P \rightarrow H$ where \ac{\ppnt{p}{i}} is a clause $C$ of
  $H$ falsified by \ppnt{p}{i}. We will call \Fi an \tb{AC-mapping}
  where ``AC'' stands for ``Assignment-to-Clause''. We will denote the
  range of \Fi as \ac{P}. (So, a clause $C$ of $H$ is in \ac{P} iff
  there is an assignment $\ppnt{p}{i} \in P$ such that $C =
  \Fi(\ppnt{p}{i})$.)
\end{definition}
\begin{definition}
 \label{def:ssa}
Let $H$ be a formula specified by a set of clauses
\s{C_1,\dots,C_k}. Let $P$ = \s{\ppnt{p}{1},\dots,\ppnt{p}{m}} be a
set of assignments to \V{H}. $P$ is called a \tb{Stable Set of
  Assignments}\footnote{In~\cite{ssp}, the notion of ``uncentered'' SSAs was introduced. The
\label{ftn:ssa}
definition of an uncentered SSA is similar to
Definition~\ref{def:ssa}. The only difference is that one requires that for
every $p_i \in P$,  $\nnbhd{p}{i}{C} \subseteq P$ holds instead of
$\NNbhd{p}{p}{i}{C} \subseteq P$.
} (SSA) of $H$ with
\tb{center} $\sub{p}{init} \in P$ if there is an AC-mapping \Fi such
that for every $\ppnt{p}{i}\in P$, $\NNbhd{p}{p}{i}{C} \subseteq P$
holds where $C = \ac{\ppnt{p}{i}}$.
\end{definition}

Note that if $P$ is an SSA of $H$ with respect to AC-mapping \Fi, then
$P$ is also an SSA of \ac{P}. 

\begin{example}
 \label{exmp:ssa}
  Let $H$ consist of four clauses: $C_1 = v_1 \vee v_2 \vee v_3$, $C_2
  = \overline{v}_1$, $C_3 = \overline{v}_2$, $C_4 = \overline{v}_3$.
  Let $P =\s{\ppnt{p}{1},\ppnt{p}{2},\ppnt{p}{3},\ppnt{p}{4}}$ where
  $\ppnt{p}{1} = 000$, $\ppnt{p}{2} = 100$, $\ppnt{p}{3} = 010$,
  $\ppnt{p}{4}=001$.  Let \Fi be an AC-mapping specified as
  $\ac{\ppnt{p}{i}} = C_i, i = 1,\dots,4$.  Since $\ppnt{p}{i}$
  falsifies $C_i$, $i=1,\dots,4$,~~\Fi is a correct AC-mapping. Set
  $P$ is an SSA of $H$ with respect to \Fi and center \sub{p}{init}=\ppnt{p}{1}. Indeed,
  \NNbhd{p}{p}{1}{C_1}=\s{\ppnt{p}{2},\ppnt{p}{3},\ppnt{p}{4}} where $C_1
  = \ac{\ppnt{p}{1}}$ and \NNbhd{p}{p}{i}{C_i} = $\emptyset$, where
  $C_i = \ac{\ppnt{p}{i}}$, $i=2,3,4$. Thus,
  $\mi{Nbhd}(\sub{p}{init},\ppnt{p}{i},\ac{\ppnt{p}{i}}) \subseteq P$,
  $i=1,\dots,4$.
\end{example}
\subsection{SSAs and  satisfiability of a formula}
\label{ssec:ssa_sat}
\begin{proposition}
  Formula $H$ is unsatisfiable iff it has an SSA.
\end{proposition}
The proof is given in Section~\ref{app:proofs} of the appendix. A
similar proposition was proved in~\cite{ssp} for ``uncentered'' SSAs
(see Footnote~\ref{ftn:ssa}).
\begin{corollary}
  Let $P$ be an SSA of $H$ with respect to PC-mapping \Fi.
  Then the set of clauses \ac{P} is unsatisfiable. Thus,
  every clause of $H \setminus \ac{P}$ is redundant.
\end{corollary}

The set of all assignments to \V{H} forms the \ti{trivial} uncentered
SSA of $H$. Example~\ref{exmp:ssa} shows a \ti{non-trivial} SSA. The
fact that formula $H$ has a non-trivial SSA $P$ is its \ti{structural}
property. That is one cannot express the fact that $P$ is an SSA of
$H$ using only the truth table of $H$. For that reason, $P$ may not be
an SSA of a formula $H'$ logically equivalent to $H$.

%
%
%
\setlength{\intextsep}{5pt}
\begin{wrapfigure}{L}{1.9in}
\small
\begin{tabbing}
aaa\=b\=cc\= dd\= \kill
$\mi{BuildPath}(H,\Fi,\sub{p}{init},\vec{s})$\{\\
\tb{\scriptsize{1}}\>  $\mi{Path} := \mi{nil}$ \\
\tb{\scriptsize{2}}\>  $\vec{p}_1 := \sub{p}{init}$\\
\tb{\scriptsize{3}}\>  $i := 1$ \\
\tb{\scriptsize{4}}\>  while ($\vec{p}_i \neq \vec{s}$) \{\\
\tb{\scriptsize{5}}\Tt   $\mi{Path} := \mi{AddAssgn}(\mi{Path},\vec{p}_i)$ \\
\tb{\scriptsize{6}}\Tt   $C := \ac{\ppnt{p}{i}}$  \\
\tb{\scriptsize{7}}*\Tt   $v := \mi{FindVar}(C,\vec{p}_i,\vec{s})$ \\
\tb{\scriptsize{8}}\Tt   $\vec{p}_{i+1} := \mi{FlipVar}(\vec{p}_i,v)$ \\
\tb{\scriptsize{9}}\Tt   $i := i+1$  \} \\
\tb{\scriptsize{10}}\>  return($\mi{Path}$) \}\\
\end{tabbing} 
\vspace{-20pt}
\caption{\ti{BuildPath} procedure}
\label{fig:bld_path}
\end{wrapfigure}

The relation between SSAs and satisfiability can be explained as
follows. Suppose that formula $H$ is satisfiable. Let \sub{p}{init} be
an arbitrary assignment to \V{H} and \pnt{s} be a satisfying
assignment that is the closest to \sub{p}{init} in terms of the
Hamming distance.  Let $P$ be the set of all assignments to \V{H} that
falsify $H$ and \Fi be an AC-mapping from $P$ to $H$.  Then \pnt{s}
can be reached from \sub{p}{init} by procedure \ti{BuildPath} shown in
Figure~\ref{fig:bld_path}. (This procedure is non-deterministic: an
oracle is used in line 7 to pick a variable to flip.) It generates a
sequence of assignments $\ppnt{p}{1},\dots,\ppnt{p}{i}$ where
\ppnt{p}{1} = \sub{p}{init} and \ppnt{p}{i}=\pnt{s}. First,
\ti{BuildPath} checks if current assignment \ppnt{p}{i} equals
\pnt{s}. If so, then \pnt{s} has been reached.  Otherwise,
\ti{BuildPath} uses clause $C=\ac{\ppnt{p}{i}}$ to generate next
assignment. Since \pnt{s} satisfies $C$, there is a variable $v \in
\V{C}$ that is assigned differently in \ppnt{p}{i} and
\pnt{s}. \ti{BuildPath} generates a new assignment \ppnt{p}{i+1}
obtained from \ppnt{p}{i} by flipping the value of $v$.

\ti{BuildPath} converges to \pnt{s} in $k$ steps where $k$ is the
Hamming distance between \pnt{p} and \pnt{s}.  Importantly,
\ti{BuildPath} reaches \pnt{s} for \ti{any} AC-mapping. Let $P$ be an
SSA of $H$ with respect to center \sub{p}{init} and AC-mapping
\Fi. Then if \ti{BuildPath} starts with \sub{p}{init} and uses \Fi as
AC-mapping, it can reach only assignments of $P$. Since every
assignment of $P$ falsifies $H$, no satisfying assignment can be
reached.

%
%
\begin{wrapfigure}{l}{2.0in}
\small
\vspace{-10pt}
\begin{tabbing}
aaa\=bb\=cc\= dd\= \kill
$\mi{BuildSSA}(H)$\{\\
\tb{\scriptsize{1}}\> $E = \emptyset$;  $\Fi := \emptyset$  \\
\tb{\scriptsize{2}}\> $\sub{p}{init} := \mi{PickInitAssgn}(H)$\\
\tb{\scriptsize{3}}\> $Q := \s{\sub{p}{init}}$   \\
\tb{\scriptsize{4}}\>  while ($Q \neq \emptyset$) \{\\
\tb{\scriptsize{5}}\Tt  $\pnt{p} := \mi{PickAssgn}(\mi{Q})$  \\
\tb{\scriptsize{6}}\Tt  $\mi{Q} := \mi{Q} \setminus \s{\pnt{p}}$ \\
\tb{\scriptsize{7}}\Tt  if $(\mi{SatAssgn}(\pnt{p},H))$ \\
\tb{\scriptsize{8}}\ttt   return($\pnt{p},\mi{nil},\mi{nil},\mi{nil}$) \\
\tb{\scriptsize{9}}\Tt $C := \mi{PickFalsifClause}(H,\pnt{p})$   \\
\tb{\scriptsize{10}}\Tt $New := \Nbhd{p}{p}{C} \setminus E$ \\
\tb{\scriptsize{11}}\Tt $Q := Q \cup New$ \\
\tb{\scriptsize{12}}\Tt $E := E \cup \s{\pnt{p}}$ \\
\tb{\scriptsize{13}}\Tt $\Fi := \Fi \cup \s{(\pnt{p},C)}$\} \\
\tb{\scriptsize{14}}\> return($\mi{nil},E,\sub{p}{init},\Fi$) \}\\
\end{tabbing} 
\vspace{-20pt}
\caption{\ti{BuildSSA} procedure}
\vspace{20pt}
\label{fig:bld_ssa}
\end{wrapfigure}

A procedure for generation of SSAs called \ti{BuildSSA} is shown in
Figure~\ref{fig:bld_ssa}. It accepts formula $H$ and outputs either a
satisfying assignment or an SSA of $H$, a center \sub{p}{init} and
AC-mapping \Fi. \ti{BuildSSA} maintains two sets of assignments
denoted as $E$ and $Q$.  Set $E$ contains the examined assignments
i.e. ones whose neighborhood is already explored.  Set $Q$ specifies
assignments that are queued to be examined. $Q$ is initialized with an
assignment \sub{p}{init} and $E$ is originally empty. \ti{BuildSSA}
updates $E$ and $Q$ in a \ti{while} loop. First, \ti{BuildSSA} picks
an assignment \pnt{p} of $Q$ and checks if it satisfies $H$. If so,
\pnt{p} is returned as a satisfying assignment. Otherwise,
\ti{BuildSSA} removes \pnt{p}~\,from $Q$ and picks a clause $C$ of $H$
falsified by \pnt{p}. The assignments of $\Nbhd{p}{p}{C}$ that are not
in $E$ are added to $Q$. After that, \pnt{p} is added to $E$ as an
examined assignment, pair $(\pnt{p},C)$ is added to \Fi and a new
iteration begins. If $Q$ is empty, $E$ is an SSA with center
\sub{p}{init} and AC-mapping \Fi.

\section{Complete Test Sets}
\label{sec:cts}
Let $N(X,Y,z)$ be a single-output combinational circuit where $X$ and
$Y$ are sets of variables specifying input and internal variables of
$N$. Variable $z$ specifies the output of $N$. Let $N$ consist of
gates $G_1,\dots,G_k$.  Then $N$ can be represented as CNF formula
$F_N = F_{G_1} \wedge \dots \wedge F_{G_k}$ where
$F_{G_i},i=1,\dots,k$ is a CNF formula specifying the consistent
assignments of gate $G_i$. Proving $N \equiv 0$ reduces to showing
that formula $F_N \wedge z$ is unsatisfiable.
\setlength{\intextsep}{4pt}
\begin{wrapfigure}{L}{1.7in}
 \begin{center}
    \includegraphics[width=1.6in,height=2.1in]{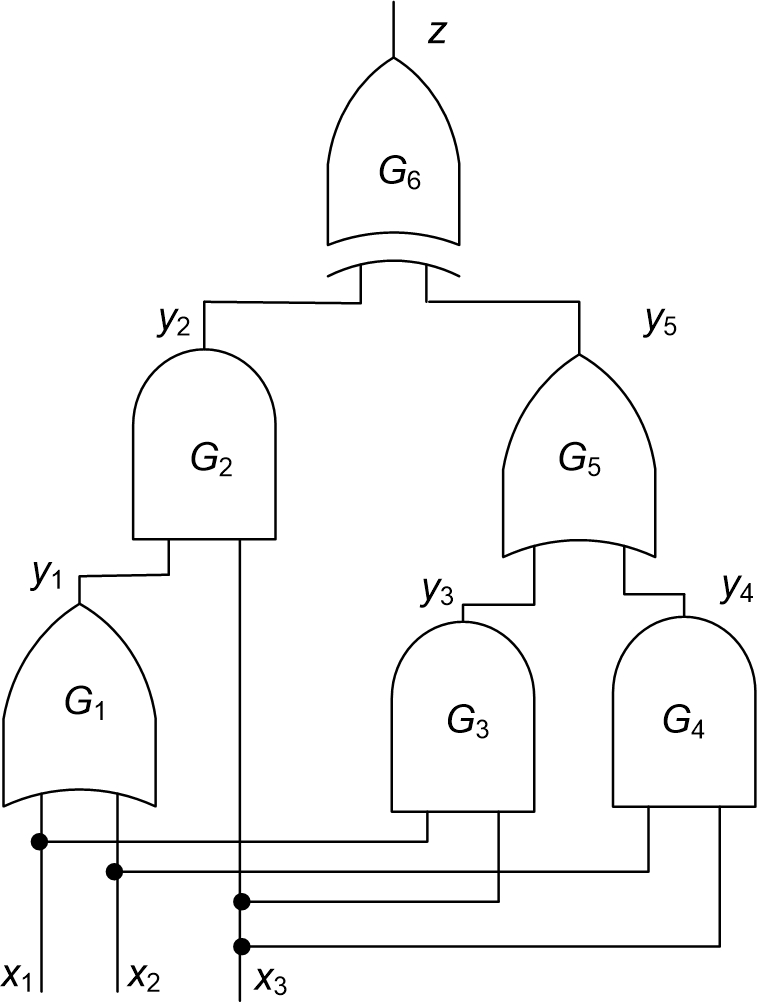}
  \end{center}
\vspace{-15pt}
\caption{Example of circuit $N(X,Y,z)$}
\vspace{-10pt}
\label{fig:miter}
\end{wrapfigure}

\begin{example}
  \label{exmp:circ}
 Circuit $N$ shown in Figure~\ref{fig:miter} represents equivalence
 checking of expressions $(x_1 \vee x_2) \wedge x_3$ and $(x_1 \wedge
 x_3) \vee (x_2 \wedge x_3)$.  The former is specified by gates $G_1$
 and $G_2$ and the latter by $G_3$, $G_4$ and $G_5$. Formula $F_N$ is
 equal to $F_{G_1} \wedge \dots \wedge F_{G_6}$ where, for instance,
 $F_{G_1} = C_1 \wedge C_2 \wedge C_3$, $C_1 = x_1 \vee x_2 \vee
 \overline{y}_1$, $C_2 = \overline{x}_1 \vee y_1$, $C_3 =
 \overline{x}_2 \vee y_1$.  Every satisfying assignment to \V{F_{G_1}}
 corresponds to a consistent assignment to gate $G_1$ and vice
 versa. For instance, $(x_1=0,x_2=0,y_1=0)$ satisfies $F_{G_1}$ and is
 a consistent assignment to $G_1$ since the latter is an OR
 gate. Formula $F_N \wedge z$ is unsatisfiable due to functional
 equivalence of expressions $(x_1 \vee x_2) \wedge x_3$ and $(x_1
 \wedge x_3) \vee (x_2 \wedge x_3)$. Thus, $N \equiv 0$.
\end{example}

Let \pnt{x} be a test i.e. an assignment to $X$.  The set of
assignments to \V{N} sharing the same assignment \pnt{x} to $X$ forms
a cube of $2^{|Y|+1}$ assignments. (Recall that $\V{N} = X \cup Y \cup
\s{z}$.)  Denote this set as \cube{x}. Only one assignment of \cube{x}
specifies the correct execution trace produced by $N$ under \pnt{x}.
All other assignments can be viewed as ``erroneous'' traces under test
\pnt{x}.
\begin{definition}
\label{def:cts}
  Let $T$ be a set of tests \s{\ppnt{x}{1},\dots,\ppnt{x}{k}} where $k
  \leq 2^{|X|}$.  We will say that $T$ is a \tb{Complete Test Set
    (CTS)} for $N$ if $\Cube{x}{1} \cup \dots \cup \Cube{x}{k}$
  contains an SSA for formula $F_N \wedge z$.
\end{definition}

If $T$ satisfies Definition~\ref{def:cts}, set $\Cube{x}{1} \cup \dots
\cup \Cube{x}{k}$ ``contains'' a proof that $N \equiv 0$ and so $T$
can be viewed as complete. If $k = 2^{|X|}$, $T$ is the \ti{trivial}
CTS. In this case, $\Cube{x}{1} \cup \dots \cup \Cube{x}{k}$ contains
the trivial SSA consisting of all assignments to \V{F_N \wedge
  z}. Given an SSA $P$ of $F_N \wedge z$, one can easily generate a
CTS by extracting all different assignments to $X$ that are present in
the assignments of $P$.

\begin{example}
  Formula $F_N \wedge z$ of Example~\ref{exmp:circ} has an SSA of 21
  assignments to \V{F_N \wedge z}. They have only 5 different
  assignments to $X = \s{x_1,x_2,x_3}$.  So the set
  \s{101,100,011,010,000} of those assignments is a CTS for $N$.
\end{example}

Definition~\ref{def:cts} is meant for circuits that are not ``too
redundant''.  Its extension to the case of high redundancy is given in
Section~\ref{app:red} of the appendix.

\section{Description Of  \sas Procedure}
\label{sec:algor}
\subsection{Motivation}

  Building an SSA can be inefficient even for a small formula. This
  makes construction of a CTS for $N$ from an SSA of $F_N \wedge z$
  impractical. We address this problem by introducing procedure called
  \sas (a short for ``Semantics and Structure''). Given formula
  $G(V,W)$, \sas generates a simpler formula $H(V)$ implied by $G$ at
  the same time trying to build an SSA for $H$. We will refer to $W$
  as the set of variables to \ti{exclude}. If \sas succeeds in
  constructing an SSA of $H$, the latter is unsatisfiable and so is
  $G$.  \sas can be applied to $F_N \wedge z$ to generate tests as
  follows. Let $V$ be a subset of \V{F_N \wedge z}. First, \sas is
  applied to construct formula $H(V)$ implied by $F_N \wedge z$ and an
  SSA of $H$. Then a set of tests $T$ is extracted from this SSA.

  The test set $T$ above can be considered as a CTS for \ti{a
    projection of circuit} $N$ on $V$. On the other hand, $T$ can be
  viewed as an \ti{approximation} of a CTS for circuit $N$, since $H(V)$ is
  essentially an abstraction of formula $F_N \wedge z$. In this paper,
  we give two examples of building a test set for $N$ from an SSA of
  $H$ generated by \sas. In the first example, $V$ is the set $X$ of
  input variables. Then an SSA found by \sas for $H(X)$ is itself a
  test set. The second example is given in Subsection~\ref{ssec:bug}
  where a ``piecewise'' construction of tests is described.

\begin{example}
  Consider the circuit $N$ of Figure~\ref{fig:miter}. Assume that $V =
  X$ where $X = \s{x_1,x_2,x_3}$ is the set of input
  variables. Application of \sas to $F_N \wedge z$ produces formula
  $H(X)= (\overline{x}_1 \vee \overline{x}_3) \wedge (\overline{x}_2
  \vee \overline{x}_3) \wedge (x_1 \vee x_2) \wedge x_3$. Besides,
  \sas generates an SSA of $H$ with center \sub{p}{init}=000 that
  consists of four assignments to $X$: \s{000,001,011,101}. (The
  AC-mapping is omitted here.) These assignments form a CTS for
  projection of $N$ on $X$ and an approximation of CTS for $N$.
\end{example}
\subsection{High-level description}

 In Figure~\ref{fig:sem_str}, we describe \sas as a recursive
 procedure.  Like DPLL-like SAT-algorithms~\cite{dpll,grasp,chaff},
 \sas makes decision assignments, runs the Boolean Constraint
 Propagation (BCP) procedure and performs branching. In particular, it
 uses decision levels~\cite{grasp}. A decision level consists of a
 decision assignment to a variable and assignments to single variables
 \ti{implied} by the former.  \sas accepts formula $G(V,W)$,
 \ti{partial} assignment \pnt{a} to variables of $W$ and index $d$ of
 current decision level. In the first call of \sas, $\pnt{a} =
 \emptyset$, $d = 0$. In contrast to DPLL, \sas keeps a subset of
 variables (namely those of $V$) \ti{unassigned}.  If $G$ is
 satisfiable, \sas outputs an assignment to $V \cup W$ satisfying $G$.
 Otherwise, it returns an SSA $P$ of formula $G$, its center and an
 AC-mapping \Fi.  The latter maps $P$ to clauses of $G$ that consist
 only of variables of $V$.  (\sas derives such clauses by
 resolution\footnote{Recall that resolution is applied to clauses $C'$ and $C''$ that have
opposite literals of some variable $w$. The result of resolving $C'$
and $C''$ on $w$ is the clause consisting of all literals of $C'$ and
$C''$ but those of $w$.
}).  Hence formula $H=\Fi(P)$
 depends only of variables of $V$.  The existence of an SSA means that
 $H$ and hence $G$ are unsatisfiable.

 We will refer to a clause $C$ of $G$ as a {\boldmath
   $V$}-\tb{clause}, if $V \cap \V{C} \neq \emptyset$ and all literals
 of $W$ of $C$ (if any) are falsified in the current node of the
 search tree by \pnt{a}. If a conflict occurs when assigning variables
 of $W$, \sas behaves as a regular SAT-solver with conflict clause
 learning.  Otherwise, the behavior of \sas is different in two
 aspects.  First, after BCP completes the current decision level, \sas
 tries to build an SSA of the set of $V$-clauses. If it succeeds in
 finding an SSA, $G$ is unsatisfiable in the current branch and \sas
 backtracks.  Thus, \sas has a ``non-conflict'' backtracking mode.
 Second, in the non-conflict backtracking mode, \sas uses a
 non-conflict learning. The objective of this learning is as
 follows. In every leaf of the search tree, \sas maintains the
 invariant that the set of current $V$-clauses is
 unsatisfiable. Suppose that a $V$-clause $C$ contains a literal of a
 variable $w \in W$ that is falsified by the current partial
 assignment \pnt{a}.  If \sas unassigns $w$ during backtracking, $C$
 stops being a $V$-clause. To maintain the invariant above, \sas uses
 resolution to produce a new $V$-clause that is a descendant of $C$
 and does not contain $w$.

  %
 %
%
\begin{wrapfigure}{l}{2.1in}
\small
\vspace{-10pt}
\begin{tabbing}
  // $V$ - set of variables to keep \\
  // $W$ - set of variables to exclude \\
 // \\
aaaa\=bb\=cc\= dd\= \kill
$\mi{SemStr}(G,\pnt{a},d)$\{\\
\tb{\scriptsize{1}}\> $(\mi{Cnfl},\pnt{a}) = \mi{RunBcp}(G,\pnt{a},d)$ \\
\tb{\scriptsize{2}}\> if ($\mi{Cnfl}$) \{\\
\tb{\scriptsize{3}}\Tt  $C := \mi{CnflCls}(G,\pnt{a},d)$\\
\tb{\scriptsize{4}}\Tt  $ G := G \cup \s{C}$\\
\tb{\scriptsize{5}}\Tt  $\pnt{v} := \mi{ArbitrAssgn}(V)$ \\
\tb{\scriptsize{6}}\Tt  return($G,\mi{nil},\s{\pnt{v}},\pnt{v},\s{(\pnt{v},C)}$)~\}\\
$---------------$ \\
\tb{\scriptsize{7}}\> $(\pnt{v},P,\sub{p}{init},\Fi) := \mi{BldSSA}(G,\pnt{a}))$\\
\tb{\scriptsize{8}}\> if ($P = \mi{nil}$)\{  \\
\tb{\scriptsize{9}}\Tt if ($|\mi{\pnt{a}}| = |W|$) \\
\tb{\scriptsize{10}}\ttt return($G,\pnt{a} \cup \pnt{v},\mi{nil},\mi{nil},\mi{nil}$) \}\\
\tb{\scriptsize{11}}\> else \{\\
\tb{\scriptsize{12}}\Tt $(G,\Fi)\!:=\!\mi{Normalize}(G,\Fi,\!P,\!\pnt{a},\!d)$\\
\tb{\scriptsize{13}}\Tt return($G,\mi{nil},P,\sub{p}{init},\Fi$) \}\\
$---------------$ \\
\tb{\scriptsize{14}}\>  $w := \mi{PickVar}(W,\pnt{a})$    \\
\tb{\scriptsize{15}}\>  $d := d + 1$ \\
\tb{\scriptsize{16}}\>  $\pnt{a}_0 := \mi{AddDecLvl}(\pnt{a},\!(w=0)\!,d)$ \\
\tb{\scriptsize{17}}\>  ($G,\pnt{s},\!P_0,\sub{p}{init},\Fi_0)\!:=\!\mi{SemStr}(G,\pnt{a}_0,d)$ \\
\tb{\scriptsize{18}}\> if ($\pnt{s} \neq \mi{nil}$) return($G,\pnt{s},\mi{nil},\mi{nil},\mi{nil}$) \\
\tb{\scriptsize{19}}\> if ($w\!\not\in \V{\Fi_0(P_0)}$) \\
\tb{\scriptsize{20}}\Tt    return($G,\mi{nil},\!P_0,\sub{p}{init},\Fi_0$) \\[4pt]
\tb{\scriptsize{21}}\>  $\pnt{a}_1 := \mi{AddDecLvl}(\pnt{a},\!(w=1)\!,d)$ \\
\tb{\scriptsize{22}}\>   ($G,\pnt{s},P_1,\Fi_1) := \mi{SemStr}(G,\pnt{a}_1,d)$ \\
\tb{\scriptsize{23}}\> if ($\pnt{s} \neq \mi{nil}$) return($G,\pnt{s},\mi{nil},\mi{nil},\mi{nil}$) \\
\tb{\scriptsize{24}}\> $H_0 := \Fi_0(P_0)$; $H_1 := \Fi_1(P_1)$; \\
\tb{\scriptsize{25}}\> $(G,\!P,\!\sub{p}{init},\Fi)\!:=\!\mi{Excl}(G,\!H_0,\!H_1,\!\pnt{a},w)$  \\
\tb{\scriptsize{26}}\> return($G,\mi{nil},P,\sub{p}{init},\Fi$) \}\\
\end{tabbing} 
\vspace{-20pt}
\caption{\ti{SemStr} procedure}
\label{fig:sem_str}
\end{wrapfigure}

  \subsection{\sas in more detail}
  As shown in Figure~\ref{fig:sem_str}, \sas consists of three parts
  separated by dotted lines. In the first part (lines 1-6), \sas runs
  BCP to fill in the current decision level number $d$. Since \sas
  does not assign variables of $V$, BCP ignores clauses that contain a
  variable of $V$. If, during BCP, a clause consisting only of
  variables of $W$ gets falsified, a conflict occurs. Then \sas
  generates a conflict clause $C$ (line 3) and adds it to $G$. In this
  case, formula $H(V)$ consists simply of $C$ that is empty (has no
  literals) in subspace specified by \pnt{a}.  Any set $P =
  \s{\pnt{v}}$ where \pnt{v} is an arbitrary assignment to $V$ is an
  SSA of $H$ in subspace specified by \pnt{a}.

  If no conflict occurs in the first part, \sas starts the second part
  (lines 7-13). Here, \sas runs \ti{BldSSA} procedure to check if the
  current set of $V$-clauses is unsatisfiable by building an SSA. If
  \ti{BldSSA} fails to build an SSA (line 8), it checks if all
  variables of $W$ are assigned (line 9). If so, formula $G$ is
  satisfiable.  \sas returns a satisfying assignment (line 10) that is
  the union of current assignment \pnt{a} to $W$ and assignment
  \pnt{v} to $V$ returned by \ti{BldSSA}.  (Assignment \pnt{v}
  satisfies all the current $V$-clauses).

  If \ti{BldSSA} succeeds in building an SSA $P$ with respect to an
  AC-function \Fi and center \sub{p}{init} (line 11), \sas performs
  operation called \ti{Normalize} over formula $H$ where $H = \ac{P}$
  (line 12). After that, \sas returns. Let $w$ be the decision
  variable of the current decision level (i.e.  level number $d$).
  The objective of \ti{Normalize} is to guarantee that every clause of
  $H$ contains no more than one variable assigned at level $d$ and
  this variable is $w$. Let $C$ be a clause of $H$ that violates this
  rule.  Suppose, for instance, that $C$ has one or more literals
  falsified by \ti{implied} assignments of level $d$. In this case,
  \ti{Normalize} performs a sequence of resolution operations that
  starts with clause $C$ and terminates with a clause $C^*$ that
  contains only variable $w$.  (This is similar to the conflict
  generation procedure of a SAT-solver. It starts with a clause
  rendered unsatisfiable that has at least two literals assigned at
  the conflict level. After a sequence of resolutions, this procedure
  generates a clause where only one literal is falsified at the
  conflict level.)  Importantly, $C^*$ and $C$ are \ti{identical} as
  $V$-clauses i.e. they are different only in literals of $W$.  Clause
  $C^*$ is added to $G$ and replaces $C$ in AC-function \Fi and hence
  in $H$.
%
  %
%
\setlength{\intextsep}{5pt}
\begin{wrapfigure}{L}{1.9in}
\small
\begin{tabbing}
aaa\=b\=cc\= dd\= \kill
$\mi{Excl}(G,H_0,H_1,\pnt{a},w)$\{\\
\tb{\scriptsize{1}}\>  $H := H_0 \cup H_1$  \\
\tb{\scriptsize{2}}\>  $H^w :=  \s{C \in H |w\!\in\! \V{C}}$\\
\tb{\scriptsize{3}}\>  $H :=\!H \setminus H^w$ \\
\tb{\scriptsize{4}}\>  while (\ti{true}) \{ \\
\tb{\scriptsize{5}}\Tt   $(\pnt{v},\!P,\!\sub{p}{init},\Fi)\!:=\!\mi{BldSSA}(H,\pnt{a})$ \\
\tb{\scriptsize{6}}\Tt   if ($P \neq \mi{nil}$) return($G,P,\sub{p}{init},\Fi$)\\
\tb{\scriptsize{7}}\Tt   $C := \mi{GenCls}(H^w,\pnt{v})$ \\
\tb{\scriptsize{8}}\Tt   $H := H \cup \s{C}$ \} \\
\tb{\scriptsize{9}}\Tt   $G := G \cup \s{C}$ \} \} \\
\end{tabbing} 
\vspace{-10pt}
\caption{\ti{Excl} procedure}
\label{fig:excl_var}
\end{wrapfigure}

  If neither satisfying assignment nor SSA is found in the second
  part, \sas starts the third part (lines 14-26) where it
  branches. First, a decision variable $w$ is picked to start decision
  level number $d+1$.  \sas adds assignment $w=0$ to \pnt{a} and calls
  itself to explore the left branch (line 17).  If this call returns a
  satisfying assignment \pnt{s}, \sas ends the current invocation and
  returns \pnt{s} (line 18).  If $\pnt{s} = \mi{nil}$ (i.e. no
  satisfying assignment is found), \sas checks if the set of clauses
  $\Fi_0(P_0)$ found to be unsatisfiable in branch $w=0$ contains
  variable $w$. If not, then branch $w=1$ is skipped and \sas returns
  SSA $P_0$,~\sub{p}{init} and AC-mapping $\Fi_0$ found in the left
  branch. Otherwise, \sas examines branch $w=1$ (lines 21-23).

  Finally, \sas merges results of both branches by calling procedure
  \ti{Excl}. Formulas $H_0$ and $H_1$ specify unsatisfiable
  $V$-clauses of branches $w=0$ and $ w=1$ respectively.  This means
  that formula $H_1 \wedge H_2$ is unsatisfiable in the subspace
  specified by \pnt{a}.  However, \sas maintains a stronger invariant
  that all $V$-\ti{clauses} are unsatisfiable in subspace
  \pnt{a}. This invariant is broken after unassigning $w$ since the
  clauses of $H_1 \wedge H_2$ containing variable $w$ are not
  $V$-clauses any more. Procedure \ti{Excl} ``excludes'' $w$ to
  restore this invariant via producing new $V$-clauses obtained by
  resolving clauses of $H_1$ and $H_2$ on $w$. 

  The pseudo-code of \ti{Excl} is shown in Figure~\ref{fig:excl_var}.
  First, \ti{Excl} builds formula $H$ that consists of clauses of $H_1
  \cup H_2$ minus those that have variable $w$ (lines 1-3). Then
  \ti{Excl} tries to build an SSA $P$ of $H$ by calling procedure
  \ti{BldSSA} in a \ti{while} loop (lines 4-9). If \ti{BldSSA}
  succeeds, \ti{Excl} returns the SSA found by \ti{BldSSA}. Otherwise,
  \ti{BldSSA} returns an assignment \pnt{v} that satisfies $H$. This
  satisfying assignment is eliminated by generating a $V$-clause $C$
  falsified by \pnt{v} and adding it to $H$. Clause $C$ is generated
  by resolving two clauses of $H_1 \cup H_2$ on variable $w$. After
  that, a new iteration begins.
\raggedbottom

\section{Example Of How \sas Operates}
Let $V=\s{v_1,v_2}$, $W=\s{w_1,w_2}$ and $G(V,W)$ be a formula of 6
clauses: 
$C_1=w_1 \vee v_1$, $C_2=w_1 \vee w_2$, $C_3=\overline{w}_2 \vee v_2$,
$C_4= \overline{v}_1 \vee \overline{v}_2$,~$C_5=\overline{w}_1 \vee v_1$,
$C_6 = \overline{w}_1 \vee v_2$.

Let us consider how \sas operates on the formula above. We will
identify invocations of \sas by partial assignment \pnt{a} to $W$. For
instance, since \pnt{a} is empty in the initial call of \sas, the
latter is denoted as $\sas_{\emptyset}$. We will also use \pnt{a} as a
subscript to identify $G$ under assignment \pnt{a}.  The first part of
$\sas_{\emptyset}$ (see Figure~\ref{fig:sem_str}) does not trigger any
action because $G_{\emptyset}$ does not contain unit clauses (i.e.
unsatisfied clauses that have only one unassigned literal).  In the
second part of $\sas_{\emptyset}$, procedure \ti{BldSSA} fails to
build an SSA because the only $V$-clause of $G_{\emptyset}$ is $C_4$.
So the current set of $V$-clauses is satisfiable. Having found out
that not all variables of $W$ are assigned (line 9 of
Figure~\ref{fig:sem_str}), $\sas_{\emptyset}$ leaves the second part.

Let $w_1$ be the variable of $W$ picked in the third part for
branching (line 14). $\sas_{\emptyset}$ uses assignment $w_1 = 0$ to
start decision level number 1. (In the original call, the decision
level value is 0). Then $\sas_{(w_1=0)}$ is invoked that operates as
follows.  $G_{(w_1=0)}$ contains unit clauses $C_1=\cancel{w_1} \vee
v_1$ and $C_2=\cancel{w_1} \vee w_2$ (we crossed out literal $w_1$ as
falsified). Unit clause $C_1$ is ignored by BCP, since \sas does not
assign variables of $V$. On the other hand, BCP assigns value 1 to
$w_2$ to satisfy $C_2$. So current \pnt{a} equals $(w_1=0,w_2=1)$ and
decision level number 1 contains one decision and one implied
assignment. At this point, BCP stops. The only clause consisting
solely of variables of $W$ (clause $C_2$) is satisfied. So no conflict
occurred and $\sas_{(w_1=0)}$ finishes the first part of the code.

Current formula $G_{(w_1=0,w_2=1)}$ has the following $V$-clauses:
$C_1=\cancel{w_1} \vee v_1$, $C_3=\cancel{\overline{w}_2} \vee v_2$,
$C_4=\overline{v}_1 \vee \overline{v}_2$. This set of $V$-clauses is
unsatisfiable. \ti{BldSSA} proves this by generating a set $P$ of
three assignments: \ppnt{v}{1}=11, \ppnt{v}{2}=01, \ppnt{v}{3}=10 that
is an SSA. The center is \ppnt{v}{1} and the AC-function \Fi is
defined as \ac{\ppnt{v}{1}} = $C_4$, \ac{\ppnt{v}{2}} = $C_1$,
\ac{\ppnt{v}{3}} = $C_3$.  So formula $H = \Fi(P)$ for subspace
\pnt{a} consists of clauses $C_1,C_3,C_4$. Note that $H$ needs
normalization, since $C_3$ contains literal $\overline{w}_2$ falsified
by the \ti{implied} assignment of level 1. Procedure \ti{Normalize}
(line 12) fixes this problem. It produces new clause $C_7 = w_1 \vee
v_2$ obtained by resolving $C_3 = \overline{w_2} \vee v_2$ with clause
$C_2 = w_1 \vee w_2$ on $w_2$.  (Note that $C_2$ is the clause from
which assignment $w_2 = 1$ was derived during BCP.) Clause $C_7$ is
added to $G$. It replaces clause $C_3$ in \Fi and hence in $H$. So now
\ac{\ppnt{v}{3}} = $C_7$ and $H$ consists of clauses $C_1,C_7,C_4$. At
this point, $\sas_{(w_1=0)}$ terminates returning SSA $P$, center
~\ppnt{v}{1}, AC-mapping ~\Fi and modified $G$ to $\sas_{\emptyset}$.

Having completed branch $w_1=0$, $\sas_{\emptyset}$ invokes
$\sas_{(w_1=1)}$. Since $G_{(w_1=1)}$ does not have any unit clauses,
no action is taken in the first part. Formula $G_{(w_1=1)}$ contains
three $V$-clauses: $C_4= \overline{v}_1 \vee \overline{v}_2$,
$C_5=\cancel{\overline{w}_1} \vee v_1$ and $C_6 =
\cancel{\overline{w}_1} \vee v_2$. Procedure \ti{BldSSA} proves them
unsatisfiable by generating a set $P$ of three assignments
\ppnt{v}{1}=11, \ppnt{v}{2}=01, \ppnt{v}{3}=10 that is an SSA with
respect to center \ppnt{v}{1} and AC-function: \ac{\ppnt{v}{1}} =
$C_4$, \ac{\ppnt{v}{2}} = $C_5$, \ac{\ppnt{v}{3}} = $C_6$.  So formula
$H = \Fi(P)$ consists of clauses $C_4,C_5,C_6$. It does not need
normalization.  $\sas_{(w_1=1)}$ terminates returning SSA $P$,
~\ppnt{v}{1}, and ~\Fi to $\sas_{\emptyset}$.

Finally, $\sas_{\emptyset}$ calls \ti{Excl} to merge the results of
branches $w_1=0$ and $w_1=1$ by excluding variable $w_1$. Formulas
$H_0$ and $H_1$ passed to \ti{Excl} specify unsatisfiable sets of
$V$-clauses found in branches $w_1=0$ and $w_1=1$ respectively. Here,
$H_0=\s{C_1,C_4,C_7}$ and $H_1=\s{C_4,C_5,C_6}$. \ti{Excl} starts by
generating formulas $H^{w_1}$ and $H$ (lines 1-3 of
Figure~\ref{fig:excl_var}).  Formula $H^{w_1}=\s{C_1,C_5,C_6,C_7}$
consists of the clauses of $H_0 \cup H_1$ with variable $w_1$. Formula
$H=\s{C_4}$ is equal to $(H_0 \cup H_1) \setminus H^{w_1}$. Then
\ti{Excl} tries to build an SSA for $H$ in a \ti{while} loop (lines
4-9). Since current formula $H$ is satisfiable, a satisfying
assignment \pnt{v} is returned by \ti{BldSSA} in the first
iteration. Assume that \pnt{v}=01.  To exclude this assignment,
\ti{Excl} generates clause $C_8= v_1$ (by resolving $C_1=w_1 \vee v_1$
of $H_0$ and $C_5=\overline{w}_1 \vee v_1$ of $H_1$ on $w_1$) and adds
it to $H$ and $G$.

$H$ is still satisfiable. Thus, the satisfying assignment $\pnt{v}=10$
is returned by \ti{BldSSA} in the second iteration. To exclude it,
clause $C_9 = v_2$ is generated (by resolving $C_7 = w_1 \vee v_2$ and
$C_6 = \overline{w}_1 \vee v_2$) and added to $H$ and $G$.  In the
third iteration, \ti{BldSSA} proves $H$ unsatisfiable by generating an
SSA $P$ of three assignments \ppnt{v}{1}=11, \ppnt{v}{2}=01,
\ppnt{v}{3}=10. Assignment \ppnt{v}{1} is the center and the
AC-function is defined as \ac{\ppnt{v}{1}} = $C_4$, \ac{\ppnt{v}{2}} =
$C_8$, \ac{\ppnt{v}{3}} = $C_9$ where $C_4 = \overline{v}_1 \vee
\overline{v}_2$, $C_8=v_1$, $C_9=v_2$.  The modified formula $G$ with
$P$,~\ppnt{v}{1} and \Fi are returned by \ti{Excl} to
$\sas_{\emptyset}$. They are also returned by $\sas_{\emptyset}$ as
the final result.

\section{Application Of \sas To Testing}
\label{sec:appl_test}
Let $M$ be a multi-output combinational circuit.  In this section, we
consider some applications of \sas to testing $M$. They can be used in
two scenarios.  The first scenario is as follows.  Let $\xi$ be a
property of $M$ specified by a single-output circuit $N$. Consider the
case where $\xi$ can be proved by a SAT-solver. If one needs to check
$\xi$ only once, using the current version of \sas does not make much
sense (it is slower than a SAT-solver). Assume however that one
frequently modifies $M$ and needs to check that property $\xi$ still
holds.  Then one can apply \sas to generate a CTS for a projection of
$N$ and then re-use this CTS as a high-quality test set every time
circuit $M$ is modified (Subsection~\ref{ssec:des_changes}).

The second scenario is as follows. Assume that some properties of $M$
cannot be solved by a SAT-solver and/or one needs to verify the
correctness of circuit $M$ ``as a whole''. (In the latter case, a
SAT-solver is typically used to construct tests generating events
required by a coverage metric.) Then tests generated by \sas can be
used, for instance, to hit corner cases more often
(Subsection~\ref{ssec:corners}) or to empower a traditional test set
with CTSs for local properties of $M$
(Subsection~\ref{ssec:loc_props}).

\subsection{Verification of design changes}
\label{ssec:des_changes}

Let $M^*$ be a circuit obtained by modification of $M$. Suppose that
one needs to check whether $M^*$ is still correct. This can be done by
checking if $M^*$ is logically equivalent to $M$. However, equivalence
checking cannot be used if the functionality of $M^*$ has been
intentionally modified. Another option is to run a test set previously
generated for $M$ to verify $M^*$. Generation of CTSs can be used to
empower this option. The idea here is to re-use CTSs generated for
testing the properties of $M$ that should hold for $M^*$ as well.

Let $\xi$ be a property of $M$ that is supposed to be true for $M^*$
too.  Let $N$ be a single-output circuit specifying $\xi$ for $M$ and
$T$ be a CTS constructed to check if $N \equiv 0$. To verify if $\xi$
holds for $M^*$, one just needs to apply $T$ to circuit $N^*$
specifying property $\xi$ in $M^*$. Of course, the fact that $N^*$
evaluates to 0 for the tests of $T$ does not mean that $\xi$ holds for
$M^*$. Nevertheless, since $T$ is specifically generated for $\xi$,
there is a good chance that a test of $T$ will break $\xi$ if $M^*$ is
buggy. In Subsection~\ref{ssec:bug}, we substantiate this intuition
experimentally.
%
%
\subsection{Verification of corner cases}
\label{ssec:corners}
\setlength{\intextsep}{4pt}
\begin{wrapfigure}{l}{1.35in}
 \begin{center}
    \includegraphics[width=1.2in]{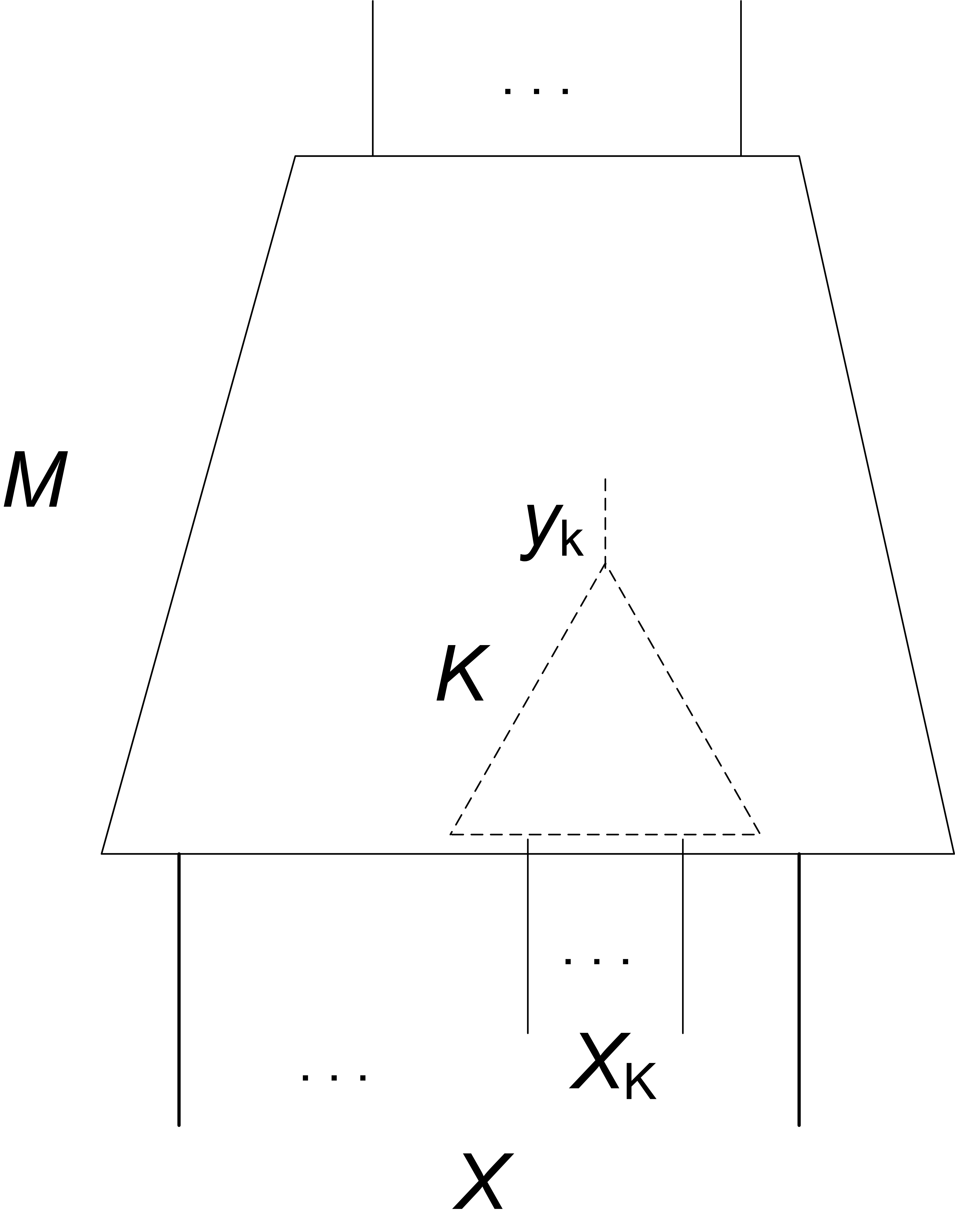}
  \end{center}
\vspace{-10pt}
\caption{Subcircuit $K$ of circuit $M$}
\vspace{30pt}
\label{fig:subcirc}
\end{wrapfigure}

Let $K$ be a single-output subcircuit of circuit $M$ as shown in
Figure~\ref{fig:subcirc}. The input variables of $K$ (set $X_K$) is a
subset of the input variables of $M$ (set $X$). Suppose that the
output of $K$ takes value 0 much more frequently then 1. Then one can
view an assignment \pnt{x} to $X$ for which $K$ evaluates to 1 as
specifying a ``corner case'' i.e. a rare event. Hitting such a corner
case even once by a random test can be very hard.  This issue can be
addressed by using a coverage metric that \ti{requires} setting the
value of $K$ to both 0 and 1.  (The task of finding a test for which
$K$ evaluates to 1, can be easily solved, for instance, by using a
SAT-solver.)  The problem however is that hitting a corner case only
once may be insufficient.

\setlength{\intextsep}{4pt}
\begin{wrapfigure}{l}{1.45in}
 \begin{center}
    \includegraphics[width=1.3in]{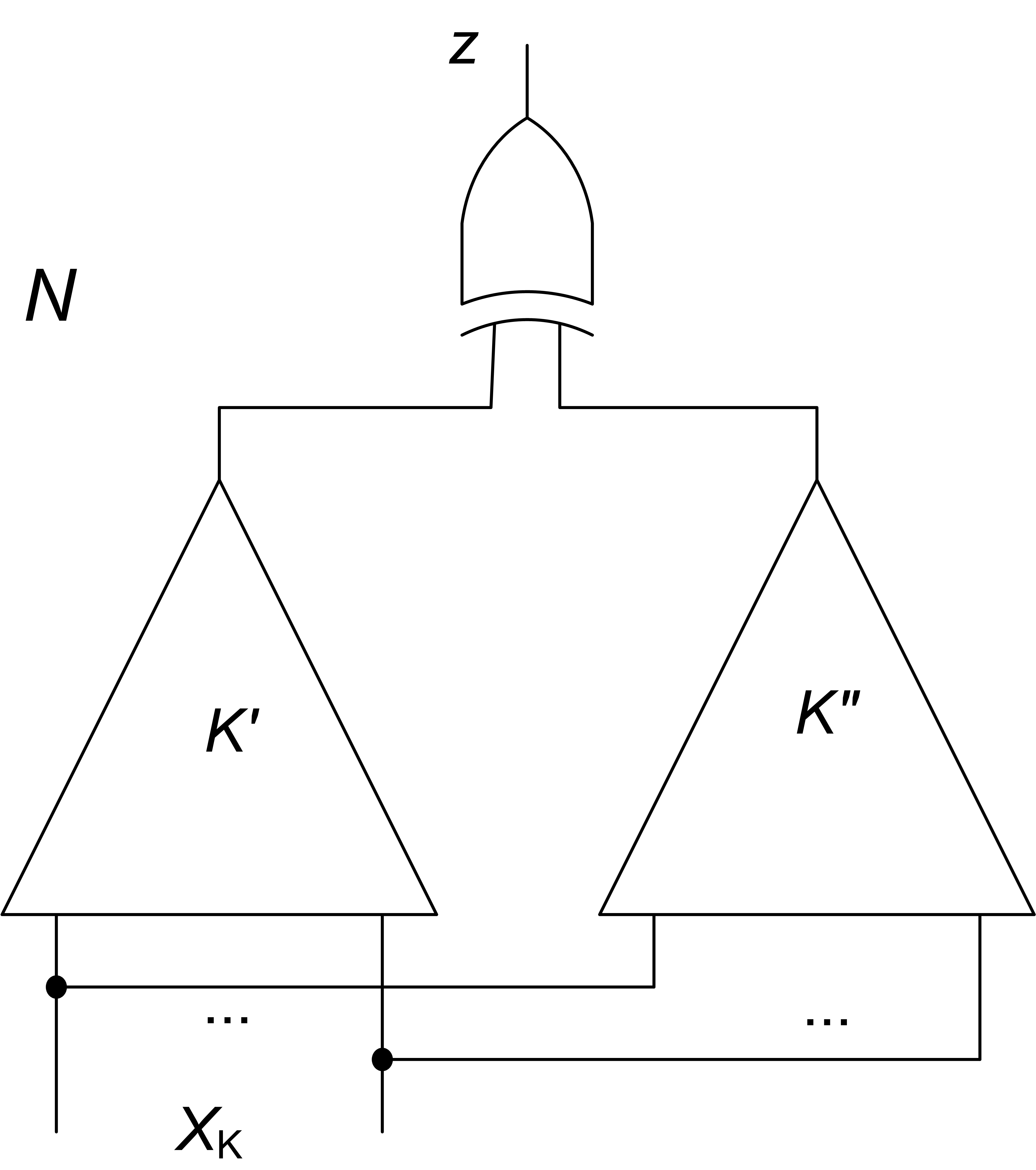}
  \end{center}
\vspace{-10pt}
\caption{The miter of circuits $K'$ and $K''$}
\label{fig:gen_miter}
\end{wrapfigure}

Ideally, it would be nice to have an option of generating a test set
where the ratio of assignments for which $K$ evaluates to 1 is higher
than in the truth table of $K$.  One can achieve this objective as
follows. Let $N$ be a miter of circuits $K'$ and $K''$ (see
Figure~\ref{fig:gen_miter}) i.e.  a circuit that evaluates to 1 iff
$K'$ and $K''$ are functionally inequivalent.  Let $K'$ and $K''$ be
two copies of circuit $K$. So $N \equiv 0$ holds. Let $T_K$ be a CTS
for projection of $N$ on $X_K$. Set $T_K$ can be viewed as a result of
``squeezing'' the truth table of $K$. Since this truth table is
dominated by assignments for which $K$ evaluates to 0, this part of
the truth table is reduced the most\footnote{One can give a more
  precise explanation of when and why using $T_K$ should work.}. So,
one can expect that the ratio of tests of $T_K$ for which $K$
evaluates to 1 is higher than in the truth table of $K$. In
Subsection~\ref{ssec:ecorners}, we substantiate this intuition
experimentally. Extending an assignment \ppnt{x}{K} of $T_K$ to an
assignment \pnt{x} to $X$ is easy e.g. one can randomly assign the
variables of $X \setminus X_K$.

\subsection{Empowering testing by adding CTSs of local properties}
\label{ssec:loc_props}
Let $\Xi = \s{\xi_1,\dots,\xi_k}$ be a set of
local\footnote{Informally, property $\xi_i$ of $M$ is ``local'' if only a fraction of
$M$ is responsible for $\xi_i$.
} properties of $M$ specified by
single-output circuits $N_1,\dots,N_k$ respectively.  Typically,
testing is used to check if circuit $M$ is correct ``as a
whole''. This notion of correctness is a conjunction of \ti{many}
properties including those of $\Xi$. Let $T$ be a test set generated
by a traditional testing procedure (e.g. driven by some coverage
metric). An obvious flaw of $T$ is that it does not guarantee that
the properties of $\Xi$ hold. This problem can be addressed by using a
formal verification procedure, e.g. a SAT-solver, to check if these
properties hold. Note, however, that proving the properties of $\Xi$
by a formal verification tool does not add any new tests to $T$ and
therefore does not make $T$ more powerful.
1
Now, assume that every property $\xi_i$ of $\Xi$ is proved by building
a CTS $T_i$ for projection of $N_i$ on its input variables. Let $T^*$
denote $T \cup T_1 \cup \dots \cup T_k$. Set $T^*$ is more powerful
than $T$ combined with proving the properties of $\Xi$ by a formal
verification tool. Indeed, in addition to guaranteeing that the
properties of $\Xi$ hold, set $T^*$ contains \ti{more tests} than $T$
and hence can identify \ti{new bugs}. In
Subsection~\ref{ssec:eloc_props}, we provide some experimental data on
using \sas to verify local properties.

\section{Application Of \sas To Sat-Solving}
\label{sec:appl_sat}
Conflict Driven Clause Learning (CDCL)~\cite{grasp,chaff} has played a
major role in boosting the performance of modern SAT-solvers. However,
CDCL has the following flaw. Suppose one needs to check satisfiability
of formula $G$ equal to $A(X,Y) \wedge B(Y,Z)$ where $|Y|$ is much
smaller than $|X|$ and $|Z|$. One can view $G$ as describing
interaction of two blocks specified by $A$ and $B$ where $Y$ is the
set of variables via which these blocks communicate.  Sets $X$ and $Z$
specify the internal variables of these blocks.  A CDCL SAT-solver
tends to produce clauses that relate variables of $X$ and $Z$ turning
$G$ into a ``one-block'' formula. This can make finding a short proof
much harder.  (Intuitively, this flaw of CDCL becomes even more
detrimental when a formula describes interaction of $n$ small blocks
where $n$ is much greater than 2.) A straightforward way to solve this
problem is to avoid resolving clauses on variables of $Y$. However, a
resolution-based SAT-solver cannot do this. A goal of a resolution
proof is to generate an empty clause, which cannot be achieved without
resolving clauses on variables of $Y$.

\sas does not have the problem above since it can just replace
resolutions on variables of $Y$ with building an SSA for clauses
depending on $Y$. Then, instead of generating an empty clause, \sas
produces an unsatisfiable formula $H(Y)$ implied by $G$. Thus, \sas
can facilitate finding good proofs.  However, \sas has another issue
to address.  Currently \sas computes SSAs ``explicitly'' i.e. in terms
of single assignments. The proof system specified by such SSAs is much
weaker than resolution. This can negate the positive effect of
preserving the structure of $G$. A potential solution of this problem
is to compute an SSA in clusters e.g. cubes of assignments where a
cube can contain an exponential number of assignments. This makes SSAs
a more powerful proof system. (For instance, in~\cite{ssp}, the
machinery of SSAs is used to efficiently solve pigeon-hole formulas
that are hard for resolution.)  Computing SSAs in clusters is far from
trivial and \sas can be used as a starting point in this line of
research.

\section{Experiments}
\label{sec:exper}

In this section, we describe results of four experiments. In the first
experiment (Subsection ~\ref{ssec:cts}), we compute CTSs for circuits
and their projections.  In Subsection~\ref{ssec:bug}, we describe the
second experiment where \sas is used for bug detection. In particular,
we introduce a method for ``piecewise'' construction of
tests. Importantly, this method has the potential of being as scalable
as SAT-solving and so could be used to generate high-quality tests for
very large circuits.  In the third experiment,
(Subsection~\ref{ssec:ecorners}) we use CTSs to test corner cases. In
the last experiment (Subsection~\ref{ssec:eloc_props}), we apply \sas
to verification of local properties.
In the first three experiments, we used miters i.e. circuits
specifying the property of equivalence checking (see
Figure~\ref{fig:gen_miter}).  In the fourth experiment, we tested
circuits specifying the property that an implication between two
formulas holds.
\subsection{A few remarks about current implementation of \sas}
\label{ssec:implem}

Let \sas be applied to $G(V,W)$ to produce a formula $H(V)$ and its
SSA. As we mentioned in Section~\ref{sec:algor}, when assigning values
to variables of $W$, \sas behaves almost like a regular SAT-solver. So
one can use the techniques employed by state-of-the-art SAT-solvers to
enhance their performance. However, to make implementation simpler and
easier to modify, we have not used those techniques in \sas. For
instance, when a variable is assigned a value (implied or decision), a
separate node of the search tree is created, no watched literals are
used to speed up BCP and so on.

Currently, \sas does not re-use SSAs obtained in the previous leafs of
the search tree. After backtracking, \sas starts building an SSA from
scratch. On the other hand, it is quite possible that, say, an SSA of
100,000 assignments generated in the right branch $w=1$ could have
been obtained by making minor changes in the SSA of the left branch
$w=0$.  Implementation of SSA re-using should boost the performance of
\sas (see Section~\ref{app:reuse} of the appendix).
\subsection{Computing CTSs for circuits and projections}
\label{ssec:cts}

The objective of the first experiment was to give examples of circuits
with non-trivial CTSs and to show that computing a CTS for a
\ti{projection} of $N$ is much more efficient than for $N$.  The miter
$N$ of circuits $M'$ and $M''$ (like the one shown in
Figure~\ref{fig:gen_miter} for circuits $K'$ and $K''$) we used in
this experiment was obtained as follows. Circuit $M'$ was a subcircuit
extracted from the transition relation of an HWMCC-10 benchmark. (The
motivation was to use realistic circuits.) For the nine miters we used
in this experiment, circuit $M'$ was extracted from nine different
transition relations.  Circuit $M''$ was obtained by optimizing $M'$
with ABC, a high-quality tool developed at UC Berkeley~\cite{abc}.

The results of the first experiment are shown in
Table~\ref{tbl:cts}. The first column of Table~\ref{tbl:cts} lists the
names of the examples. The second and third columns give the number of
input variables and that of gates in $N$. The following group of three
columns provide results of computing a CTS for $N$. This CTS was
obtained by applying \sas to formula $F_N \wedge z$ with an empty set
of variables to exclude. In this case, the resulting formula $H$ is
equal to $F_N \wedge z$ and \sas just constructs its SSA. The first
column of this group gives the size of the SSA found by \sas.  The
second column shows the number of different assignments to $X$ in the
assignments of this SSA. (Recall that $X$ is the set of input
variables of $N$.) The third column of this group gives the run time
of \sas. The last two columns of Table~\ref{tbl:cts} describe results
of computing CTS for a projection of $N$ on $X$. We will denote this
projection by {\boldmath $N^X$}. This CTS is obtained by applying \sas
to $F_N \wedge z$ using $Y \cup z$ as the set of variables to exclude
(where $Y$ specifies the set of internal variables of $N$).  The first
column of the two gives the size of the SSA generated for formula
$H(X)$ by \sas. The second column shows the run time of \sas.

%
%
%
\begin{wraptable}{l}{2.65in}
\small
\caption{\ti{CTSs for circuits and their projections}}
\vspace{-5pt}
\scriptsize
\begin{center}
\begin{tabular}{|c|c|c|c|c|c|c|c|} \hline
name    & \#inp\_ & \#ga-   & \multicolumn{3}{c|}{CTS for original} & \multicolumn{2}{c|}{CTS  for }  \\
        &  vars   & tes      & \multicolumn{3}{c|}{circuit} & \multicolumn{2}{c|}{projection}  \\ \cline{4-8}
        &         &           & \#SSA        & \#tests       &  time       & \#tests   &  time  \\ 
        &         &           &              &               &  (s.)       &           & (s.) \\ \hline
ex1     &  12     &  54       & 125,734      &  500          &   0.3       &  28       &  0.01   \\ \hline
ex2     &  14     &  59       & 262,405      &  3,231        &   0.6       &  1,101    &  0.04 \\ \hline
ex3     &  16     &  53       & 438,985      &  7,211        &   1.0       &  867      &  0.01 \\ \hline
ex4     &  16     &  63       & 3,265,861    &  15,868       &   9.4       &  1,452    &  0.02 \\ \hline
ex5     &  17     &  66       & 94,424       &  952          &   0.3       &  137      &  0.01 \\ \hline
ex6     &  40     &  117      &  memout      &    $*$        &   $*$       &  589      &  0.02 \\ \hline
ex7     &  40     &  454      &  memout      &    $*$        &   $*$       & 112,619   &  5.9 \\ \hline
ex8     &  50     &  317      &  memout      &    $*$        &   $*$       & 211,650   &  4.1 \\ \hline
ex9     &  55     &  215      &  memout      &    $*$        &   $*$       & 6,267     &  0.1 \\ \hline
\end{tabular}                
\end{center}
\vspace{-5pt}
\label{tbl:cts}
\end{wraptable}

For circuits ex1,..,ex5, \sas managed to build non-trivial CTSs for
the original circuits.  Their size is much smaller than $2^{|X|}$.
For instance, the trivial CTS for ex5 consists of $2^{17}$=131,072
tests, whereas \sas found a CTS of 952 tests. (So, to prove $M'$ and
$M''$ equivalent it suffices to run 952 out of 131,072 tests.) For
circuits ex6,..,ex9, \sas failed to build a non-trivial CTS due to
memory overflow. On the other hand, \sas built a CTS for projection
$N^X$ for all nine examples. Table~\ref{tbl:cts} shows that finding a
CTS for $N^X$ takes much less time than for $N$. In
Subsection~\ref{ssec:bug}, we demonstrate that although a CTS for
$N^X$ is only an approximation of a CTS for $N$, it makes a
high-quality test set.

\subsection{Using CTSs to detect bugs}
\label{ssec:bug}

%
%
%
\begin{wraptable}{L}{2.6in}
\small
\caption{\ti{Bug detection}}
\vspace{-5pt}
\scriptsize
\begin{center}
\begin{tabular}{|c|c|c|c|c|c|c|c|} \hline
name    & \#inp\_ & \#ga-   & \multicolumn{2}{c|}{random       } & \multicolumn{3}{c|}{test generation}  \\
        &  vars   & tes     & \multicolumn{2}{c|}{testing} & \multicolumn{3}{c|}{by \sas}  \\ \cline{4-8}
        &         &         & \#tests            & time          &  stra-    & \#tests   &  time  \\ 
        &         &         &  $ \times\!10^6$    & (s.)          &  tegy     &          & (s.) \\ \hline
 ex10   &  37     & 73      &  $>100$             & 181           &  1        & 254      &  0.02   \\ \hline
 ex11   &  39     & 155     &  $>100$             & 466           &  1        & 1,742    &  0.1    \\ \hline
 ex12   &  41     & 591     &  $>100$             & 826           &  1        & 25,396   &  2.2    \\ \hline
 ex13   &  42     & 307     &  $>100$             & 725           &  2        & 4,021    &  1.1    \\ \hline
 ex14   &  50     & 217     &  $>100$             & 489           &  2        & 10,147   &  7.2    \\ \hline
 ex15   &  50     & 249     &  $>100$             & 1,290         &  1        & 41,048   &  1.3    \\ \hline
 ex16   &  52     & 1,003   &  $>100$             & 707           &  2        & 707,589  &  106    \\ \hline
 ex17   &  67     & 405     &  $>100$             & 2,194         &  2        & 2,281    &  1.7    \\ \hline   
 ex18   &  70     & 265     &  $>100$             & 1,312         &  2        & 5,413    &  0.7    \\ \hline
\end{tabular}                
\end{center}
\vspace{-5pt}
\label{tbl:flt_tst}
\end{wraptable}

In the second experiment, we used \sas to generate tests exposing
inequivalence of circuits. Let $N^*$ denote the miter of circuits $M'$
and $M''$ where $M''$ is obtained from $M'$ by introducing a bug.
(Similarly to Subsection~\ref{ssec:bug}, $M'$ was extracted from the
transition relation of a HWMCC-10 benchmark and for the nine examples
of Table~\ref{tbl:flt_tst} below we used nine different transition
relations.) Denote by $N$ the miter of circuits $M'$ and $M''$ where
$M''$ is just a copy of $M'$.  In this experiment, we applied the idea
of Subsection~\ref{ssec:des_changes}: reuse the test set $T$ generated
to prove $N \equiv 0$ to test if $N^* \equiv 0$ holds. To run a single
test \pnt{x}, we used Minisat 2.0~\cite{minisat,minisat2.0}. Namely,
we added unit clauses specifying \pnt{x} to formula $F_{N^*} \wedge z$
and checked its satisfiability.

To generate $T$ we used two strategies. In strategy 1, $T$ was
generated as a CTS for projection $N^X$. Strategy 2 was employed when
\sas failed to build a CTS for $N^X$ due to memory overflow or
exceeding a time limit.  In this case, we partitioned $X$ into subsets
$X_1,\dots,X_k$ and computed sets $T_1,\dots,T_k$ where $T_i$ is a CTS
for projection $N^{X_i}$. (In the examples where we used strategy 2,
the value of $k$ was 2 or 3). The Cartesian product $T_1 \times \dots
\times T_k$ forms a test set for $N$.  Instead of building the entire
set $T$, we randomly generated tests of $T$ one by one as follows.
The next test \pnt{x} of $T$ to try was formed by taking the union of
$\pnt{x}_i$,$i=1,\dots,k$ randomly picked from corresponding
$T_i$,$i=1,\dots,k$. Note that in the extreme case where every $X_i$
consists of one variable, strategy 2 reduces to generation of random
tests. Indeed, let $X_i=\s{x_i}$,$i=1,\dots,k$ where $k = |X|$.  Then
formula $H(X_i)$ for projection $N^{X_i}$ is equal to $x_i \wedge
\overline{x}_i$. The only SSA for $H(X_i)$ is trivial and consists of
assignments $x_i=0$ and $x_i=1$ (and so does $T_i$). By randomly
choosing a test of $T_i$ one simply randomly assigns 0 or 1 to $x_i$.

\setlength{\intextsep}{4pt}
\begin{wrapfigure}{l}{1in}
 \begin{center}
   \includegraphics[width=0.8in,height=1.3in]{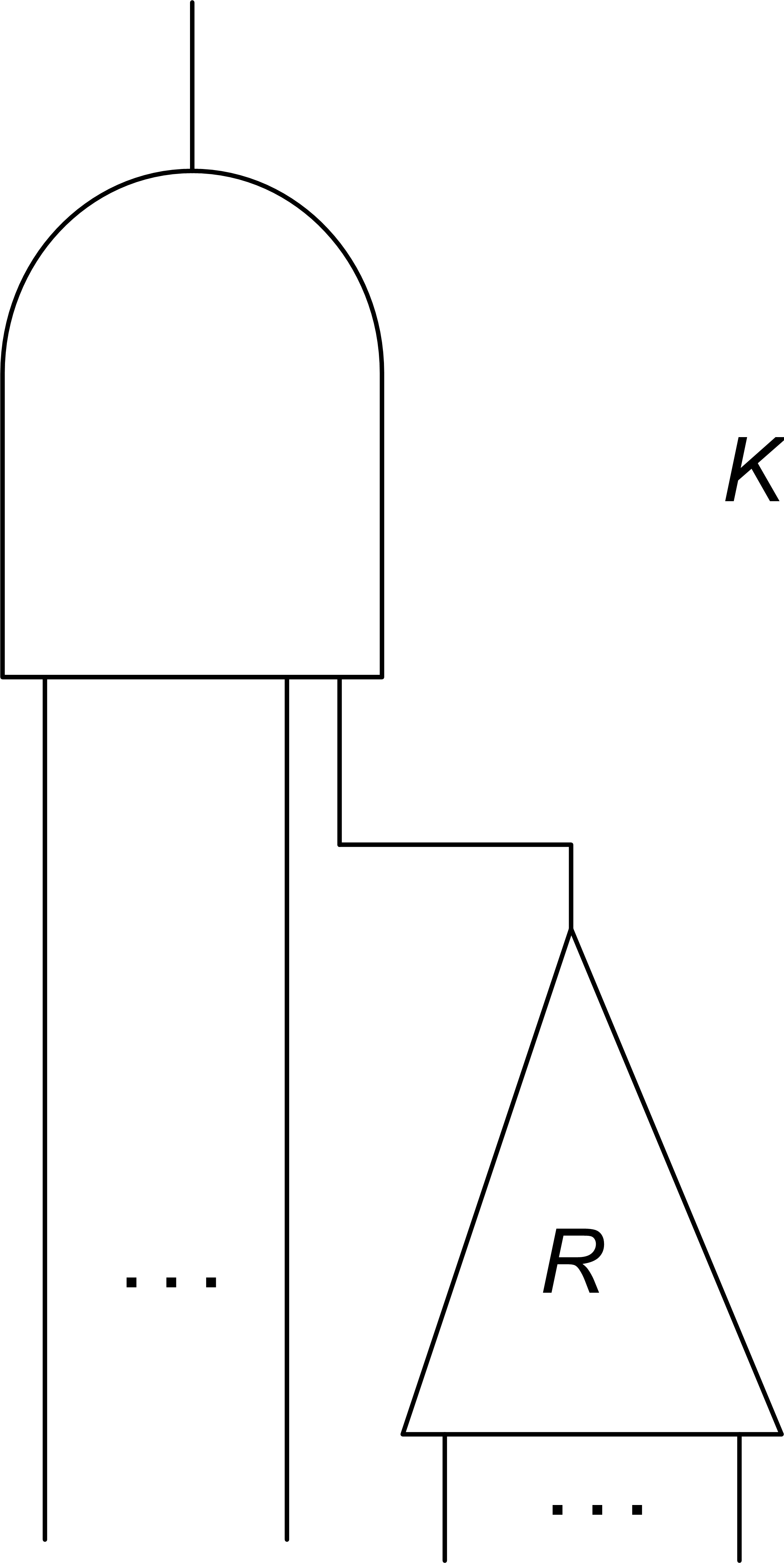}
  \end{center}
\vspace{-10pt}
\caption{Circuit $K$ whose output value is biased to 0}
\label{fig:corners}
\end{wrapfigure}

We compared our approach with random testing on small circuits. Our
objective was to show that although random testing is much more
efficient (test generation is very cheap), testing based on CTSs is
much more \ti{effective}. The majority of faults we tried was easy for
both approaches. In Table~\ref{tbl:flt_tst}, we list some examples
that turned out to be hard for random testing. The first three columns
are the same as in Table~\ref{tbl:cts}. The next two columns describe
the performance of random testing: the number of tests we tried (in
millions) and the time taken by Minisat to run all tests. The last
three columns describe the performance of our approach. The first
column of these three shows whether strategy 1 or 2 was used. The
second column gives the number of tests from $T$ one needed to run
before finding a bug. (Thus, this number is smaller than $|T|$.) The
last column of these three shows the total run-time that consists of
the time taken by \sas to generate $T$ and the time taken by Minisat
to run tests.

Table~\ref{tbl:flt_tst} shows that tests extracted from CTSs for
projections of $N$ are very effective. The fact that these tests are
effective even for strategy 2 is very encouraging for the following
reason. Computing a CTS for a projection $N^V$ where $V$ is small is
close to regular SAT-solving. (They become identical if $V =
\emptyset$.) Implementation of improvements mentioned in
Subsection~\ref{ssec:implem} should make computing a CTS for $N^V$
almost as scalable as SAT-solving. Thus, by breaking $X$ into
relatively small subsets $X_1,\dots,X_k$ and using piecewise
construction of tests as described above, one will get an effective
test set that can be efficiently computed even for very large
circuits.

\subsection{Using CTSs to check corner cases}
\label{ssec:ecorners}
%
%
%
\begin{wraptable}{L}{2.9in}
\small
\caption{\ti{Using CTSs for checking corner cases}}
\vspace{-5pt}
\scriptsize
\begin{center}
\begin{tabular}{|l|c|c|c|c|c|c|c|c|c|} \hline
name    & \#inp\_ & and & \#ga-   & \multicolumn{3}{c|}{random testing} & \multicolumn{3}{c|}{test generation}  \\
        &  vars   & inps & tes    & \multicolumn{3}{c|}{} & \multicolumn{3}{c|}{by \sas}  \\ \cline{5-10}
        &         &      &        & \#te-    & \#hi-  &   time             & \#te- & \#hits  &  time  \\ 
        &         &      &        & sts      & ts        & (s.)               & sts   &        & (s.) \\ \hline
ex19    &  50     & 10   & 72     &   $10^5$   &  54     &  0.6               &   832   &  51    &  0.03  \rule{0pt}{2.8mm}  \\ \hline
ex19*   &  60     & 20   & 72     &   $10^7$   &  0      &  65                &   1,803 &  207   &  0.1  \rule{0pt}{2.8mm}   \\ \hline
ex20   &  50     &  10  &  160   &    $10^5$  &   5     &  1.3               & 21,496  & 1,303  &  0.4  \rule{0pt}{2.8mm}  \\ \hline
ex20*  &  60     &  20  &  160   &    $10^7$  &   0     &  129               & 161,195 & 10,036 & 3.1  \rule{0pt}{2.8mm}   \\ \hline
ex21   &  65     &  10  &  108   &    $10^5$  &  68     &  0.8               & 49,947  & 4,168  &  1.2  \rule{0pt}{2.8mm}    \\ \hline
ex21*  &  75     &  20  &  108   &    $10^7$  &  0      &  81                & 44,432  & 3,528  &  1.2  \rule{0pt}{2.8mm}    \\ \hline
ex22   &  51     &  10  &  296   &    $10^5$  &  81     &  1.8               & 50,388  & 4,560  &  4.9  \rule{0pt}{2.8mm}    \\ \hline
ex22*  &  61     &  20  &  296   &    $10^7$  &  0      &  184               & 235,452 & 22,326 &  26   \rule{0pt}{2.8mm}    \\ \hline
ex23   &  60     &  10  &  125   &    $10^5$  &  43     &  1.2               & 6,834   & 259    &  0.2  \rule{0pt}{2.8mm}    \\ \hline
ex23*  &  70     &  20  &  125   &    $10^7$  &  0      &  122               & 21,083  & 1,807  &  0.4  \rule{0pt}{2.8mm}    \\ \hline
\end{tabular}                
\end{center}
\vspace{-5pt}
\label{tbl:corners}
\end{wraptable}

In the third experiment, we used CTSs to test corner cases (see
Subsection~\ref{ssec:corners}). First we formed a circuit $K$ that
evaluates to 0 for almost all input assignments. So the input
assignments for which $K$ evaluates to 1 specify ``corner cases''.
Then we compared the frequency of hitting the corner cases of $K$ by
random testing and by tests of a set $T$ built by \sas. The test set
$T$ was obtained as follows. Let $N$ be the miter of copies $K'$ and
$K''$ (see Figure~\ref{fig:gen_miter}). Set $T$ was generated as a CTS
for the projection of $N$ on its input variables.

Circuit $K$ was formed as follows. First, we extracted a circuit $R$
as a subcircuit of a transition relation (as described in the
previous subsections). Then we formed circuit $K$ by composing an
n-input AND gate and circuit $R$ as shown in Figure~\ref{fig:corners}.
Circuit $K$ outputs 1 only if $R$ evaluates to 1 and the first $n-1$
inputs variables the AND gate are set to 1 too. So the input
assignments for which $K$ evaluates to 1 are ``corner cases''.

The results of our experiment are given in
Table~\ref{tbl:corners}. The first column specifies the name of an
example. The next two columns give the total number of input variables
of $K$ and the number of input variables in the multi-input AND gate
(see Figure~\ref{fig:corners}). The next three columns describe the
performance of random testing. The first column of the three gives the
total number of tests. The next column shows the number of times
circuit $K$ evaluated to 1 (i.e. a corner case was hit). The last
column of the three gives the total run time. The last three columns
of Table~\ref{tbl:corners} describe the results of \sas. The first
column of the three shows the size of a CTS generated as described
above. The next column gives the number of times a corner case was
hit. The last column shows the total run time (that also includes the
time used to generate the CTS).

The examples of Table~\ref{tbl:corners} were generated in pairs that
shared the same circuit $R$ and were different only the size of the
AND gate (see Figure~\ref{fig:corners}). For instance, in ex19 and
ex19* we used 10-input and 20-input AND gates
respectively. Table~\ref{tbl:corners} shows that for circuits with
10-input AND gates, random testing was able to hit corner cases but
the percentage of those events was very low. For instance, for ex19,
only for 0.05\% of tests the output value of $K$ was 1 (54 out of
$10^5$ tests). The same ratio for tests generated by \sas was 6.12\%
(51 out of 832 tests). A significant percentage of tests generated by
\sas hit corner cases even in examples with 20-input AND gates in
sharp contrast to random testing that failed to hit a single corner
case.

\subsection{Using CTSs to verify local properties}
\label{ssec:eloc_props}
In the last experiment, we used \sas to build CTSs for local
properties (see Subsection~\ref{ssec:loc_props}).  Our objective here
was just to show that even the current implementation of \sas was
powerful enough to generate CTSs for local properties of non-trivial
circuits.

\begin{wraptable}{L}{2.8in}
  \small
\caption{\ti{Tests for local properties}}
\scriptsize
\vspace{-5pt}
\begin{center}
\begin{tabular}{|l|c|c|c|c|c|c|} \hline
HWMCC-10             & \#inp\_ &\#lat- &\#gates & $|C|$ & \#tests & time     \\
benchmark            & vars    & ches  &        &        &         &  s.       \\ \hline
\ti{nusmvbrp}         &  11     & 52    & 518    &  3     & 8,690   &  0.7    \\ \hline
\ti{cmugigamax}     &  34     & 29    & 646    &  4     & 1,158   &  0.2    \\ \hline
\ti{kenoopp1}         &  49     & 51    & 619    &  2     & 84      &  0.5     \\ \hline
\ti{kenflashp01}    &  61     & 57    & 1,292  &  7     & 46      &  0.9    \\ \hline
\ti{nusmvguidancep1}  &  84     & 86    & 1,823  &  3     & 767     &  1.2   \\ \hline
\ti{visprodcellp01}   &  30     &  78   & 2,807  &  2     & 534     &  1.4 \\ \hline
\ti{pdtswvroz10x6p1}  &  7      & 81    & 3,088  &  4     & 76      &  0.1 \\ \hline
\ti{pdtvissoap2}      &  21     & 205   & 4,333  &  2     & 6,408   &  1.6 \\ \hline
\ti{pdtvissfeistel}   &  68     & 361   & 9,976  &  2     & 5,078   &  0.1 \\ \hline

\end{tabular}                
\end{center}
\label{tbl:loc_props}
\vspace{-5pt}
\end{wraptable}

In the experiment, we tested local properties defined as follows.  Let
$M_T$ be a combinational circuit specifying a transition relation
$T(X,S,Y,S')$.  Here $S$ and $S'$ are sets of the present and next
state variables, and $X$ and $Y$ are sets of the combinational input
and internal variables respectively. So $X \cup S$ and $S'$ specify
the input and output variables of $M_T$ respectively.  Let $P$ be a
set of clauses specifying an inductive invariant for $T$. That is
$P(S) \wedge T \rightarrow P(S')$. Let $C$ be a clause of $P$. Then
$P(S) \wedge T \rightarrow C(S')$.  This implication can be viewed as
a \ti{property} of circuit $M_T$.  We will refer to it as a property
specified by clause $C$ (and predicate $P$). It
states\footnote{\label{foot:loc_prop}
Let $N$ be the circuit obtained by composing $M_T$ and a $|C|$-input
AND gate representing the negation of $C$. Then $N$ evaluates to 1 iff
the output of $M_T$ falsifies $C$. Proving $P(S) \wedge T \rightarrow
C(S')$ reduces to showing that $N \equiv 0$ for every input assignment
satisfying $P$. This is a variation of the problem we consider in this
paper (i.e. checking if $N \equiv 0$ holds). Fortunately, this
variation of the original problem can be solved by \sas.

} that for every input assignment
satisfying $P$, the output assignment of $M_T$ satisfies $C$.
Typically, $C$ is a short clause i.e. the number of literals of $C$ is
much smaller than $|S'|$. If only a small part of $M_T$ feeds the
output variables present in $C$, then the property specified by $C$ is
\ti{local}.

Table~\ref{tbl:loc_props} shows the results of our experiment. The
first column gives the name of an HWMCC-10 benchmark specified by
$M_T$. The next three columns show the number of input combinational
variables, state variables and gates in $M_T$. The next column gives
the number of literals of clause $C$ randomly picked from an inductive
invariant (generated by IC3~\cite{ic3}). The last two columns describe
the results of \sas in building a CTS for a projection of circuit $N$
defined in Footnote~\ref{foot:loc_prop} on the set of input variables
(i.e. on $X \cup S$).  These columns describe the size of the CTS and
the run time taken by \sas to build it. Table~\ref{tbl:loc_props}
shows that \sas managed to build CTSs for local properties of
non-trivial circuits (e.g. for circuit \ti{pdtvissfeistel} that has
9,976 gates and 361 latches).

\section{Background}
\label{sec:background}
As we mentioned earlier, the objective of applying a test to a circuit
is typically to check if the output assignment produced for this test
is correct. This notion of correctness usually means satisfying the
conjunction of \ti{many} properties of this circuit. For that reason,
one tries to spray tests uniformly in the space of all input
assignments. To avoid generation of tests that for some reason should
be or can be excluded, a set of constraints can be
used~\cite{cnst_rand}. Another way to improve the effectiveness of
testing is to run many tests at once as it is done in symbolic
simulation~\cite{SymbolSim}. Our approach is different from those
above in that it is ``property-directed'' and hence can be used to
generate property-specific tests.

The method of testing introduced in~\cite{bridging} is based on the
idea that tests should be treated as a ``proof encoding'' rather than
a sample of the search space. (The relation between tests and proofs
have been also studied in software verification, e.g.
in~\cite{UnitTests,godefroid,Beckman}). A flaw of this approach is
that testing is treated as a second-class citizen whose quality can be
measured only by a formal proof it encodes. In this paper, we take a
different point of view where testing becomes the \ti{part} of a formal
proof that performs structural derivations.

In~\cite{ken03}, it was shown that Craig's interpolation~\cite{craig}
can be used in model checking. An efficient procedure for extraction
of an interpolant from a resolution proof was given
in~\cite{pudlak,ken03}. A flaw of this procedure is that the size of
this interpolant strongly depends on the quality of the proof. As we
mentioned in Section~\ref{sec:appl_sat}, \sas offers a new way to solve
formulas with structure.  In particular, \sas can be used to compute
interpolants. Let formula $G(X,Y,Z)$ be equal to $A(X,Y) \wedge
B(Y,Z)$ and one applies \sas to solve formula $G$ by excluding the
variables of $X \cup Z$. Then formula $H(Y)$ produced from $G$ by \sas
can be represented as $H_1 \wedge H_2$ where $H_1$ and $H_2$ are
interpolants for $A$ and $B$ respectively. That is $A \rightarrow H_1
\rightarrow \overline{B}$ and $B \rightarrow H_2 \rightarrow
\overline{A}$. (This is due to the fact that \sas forbids resolutions
on variables of $Y$.)  An advantage of \sas is that it takes into
account formula structure and hence can potentially produce
high-quality interpolants.  However, currently, using \sas for
interpolant generation does not scale as well as extraction of an
interpolant from a proof.

Reasoning about SAT in terms of random walks was pioneered
in~\cite{rand_walk}. The centered SSAs we introduce in this paper bear
some similarity to sets of assignments generated in de-randomization
of Sch\"oning's algorithm~\cite{balls}. Typically, centered SSAs are
much smaller than uncentered SSAs introduced in~\cite{ssp}. A big
advantage of the uncentered SSA though is that its definition
facilitates computing an SSA in clusters of assignments (rather than
single assignments).


\section{Conclusion}
We consider the problem of finding a Complete Test Set (CTS) for a
combinational circuit $N$ that is a test set proving that $N \equiv
0$.  We use the machinery of stable sets of assignments to derive
non-trivial CTSs i.e. ones that do not include all possible input
assignments. The existence of non-trivial CTSs implies that it is more
natural to consider testing as structural rather than semantic
derivation (the former being derivation of a property that cannot be
expressed in terms of the truth table). Since computing a CTS for the
entire circuit $N$ is impractical, we present a procedure called \sas
that computes a CTS for a projection of $N$ on a subset of its
variables. The importance of \sas is twofold. First, it can be used
for generation of effective test sets. In particular, we describe a
procedure for ``piecewise'' construction of tests that can be
potentially applied to very large circuits. Second, \sas can be used
as a starting point in designing verification tools that efficiently
combine structural and semantic derivations.

\bibliographystyle{plain}
\bibliography{short_sat,local}
\vspace{15pt}
\appendix
\noindent{\large \tb{Appendix}}
\section{Proofs}
\label{app:proofs}
\setcounter{proposition}{0}
\begin{proposition}
 Formula $H$ is unsatisfiable iff it has an SSA.
\end{proposition}
\begin{proof} \tb{If part.}  Assume the contrary. Let $P$ be an SSA of $H$ with
center \sub{p}{init} and $H$ is satisfiable. Let \pnt{s} be an
assignment satisfying $H$. Let \pnt{p} be an assignment of $P$ that is
the closest to \pnt{s}~\,in terms of the Hamming distance.  Let $C
= \Fi(\pnt{p})$. Since \pnt{s} satisfies clause $C$, there is a
variable $v \in \V{C}$ that is assigned differently in \pnt{p}
and \pnt{s}.  Let \pnt{p^*} be the assignment obtained from \pnt{p} by
flipping the value of $v$.  Note that
$\pnt{p^*} \in \mi{Nbhd}(\sub{p}{init},\pnt{p},C)$.

Assume that $\pnt{p^*} \in P$. In this case, \pnt{p^*} is closer
to \pnt{s} than \pnt{p} and we have a contradiction. Now, assume that
$\pnt{p^*} \not\in P$. In this case,
$\mi{Nbhd}(\sub{p}{init},\pnt{p},C) \not\subseteq P$ and so set $P$ is
not an SSA. We again have a contradiction.

\vspace{4pt}
\noindent\tb{Only if part}. Assume that formula $H$ is unsatisfiable.  
By applying \ti{BuildSSA} shown in Figure~\ref{fig:bld_ssa} to $H$,
one generates a set $P$ that is an SSA of $H$ with respect to some
center \sub{p}{init} and AC-mapping \Fi.
\end{proof}

\section{CTSs And Circuit Redundancy}
\label{app:red}
\setlength{\intextsep}{4pt}
\begin{wrapfigure}{L}{1.1in}
 \begin{center}
 \includegraphics[width=1in]{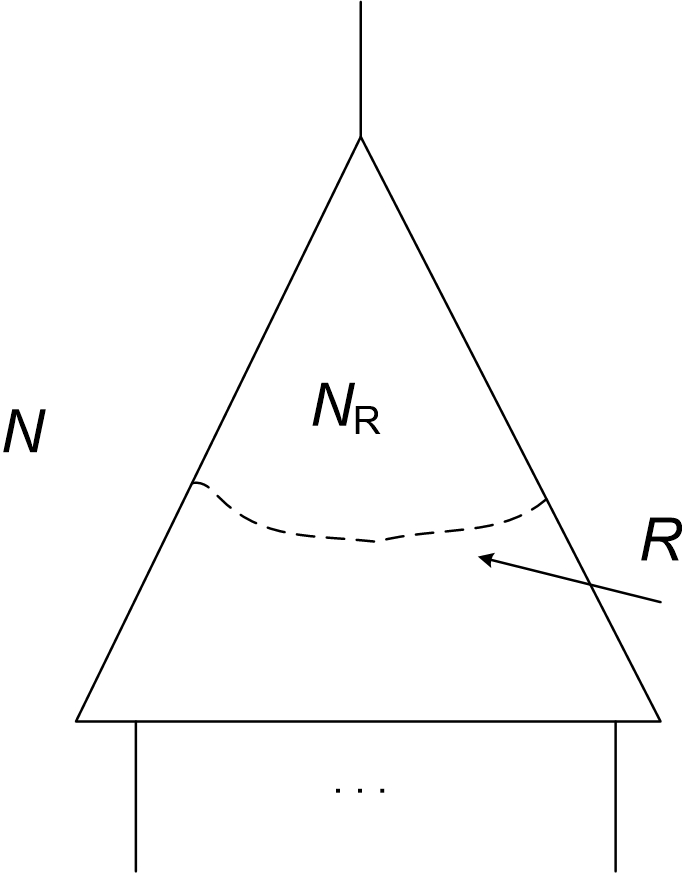}
  \end{center}
\vspace{-10pt}
\caption{A cut $R$ in circuit $N$}
\vspace{5pt}
\label{fig:cut}
\end{wrapfigure}

Let $N \equiv 0$ hold. Let $R$ be a cut of circuit $N$. We will denote
the circuit between the cut and the output of $N$ as $N_R$ (see
Figure~\ref{fig:cut}). We will say that $N$ is \tb{non-redundant} if
$N_R \not\equiv 0$ for any cut $R$ other than the cut specified by
primary inputs of $N$.

Definition~\ref{def:cts} of a CTS may not work well if $N$ is highly
redundant.  Assume, for instance, that $N_R \equiv 0$ holds for cut
$R$. This means that the clauses specifying gates of $N$ below cut $R$
(i.e. ones that are not in $N_R$) are redundant in $F_N \wedge
z$. Then one can build an SSA $P$ for $F_N \wedge z$ as follows.  Let
$P_R$ be an SSA for $F_{N_R} \wedge z$.  Let \pnt{v} be an arbitrary
assignment to the variables of $\V{N} \setminus \V{N_R}$.  Then by
adding \pnt{v} to every assignment of $P_R$ one obtains an SSA for
$F_N \wedge z$. This means that for any test \pnt{x}, \cube{x}
contains an SSA of $F_N \wedge z$. Therefore, according to
Definition~\ref{def:cts}, circuit $N$ has a CTS consisting of just one
test.

The problem above can be solved using the following observation. Let
$T$ be a set of tests \s{\ppnt{x}{1},\dots,\ppnt{x}{k}} for $N$ where
$k \leq 2^{|X|}$. Denote by $\vec{r}_i$ the assignment to the
variables of cut $R$ produced by $N$ under input \ppnt{x}{i}. Let
$T_R$ denote \s{\ppnt{r}{1},\dots, \ppnt{r}{k}}. Denote by $T^*_R$ the
set of assignments to variables of $R$ that cannot be produced in $N$
by any input assignment.  Now assume that $T$ is constructed so that
$T_R \cup T^*_R$ is a CTS for circuit $N_R$. This does not change
anything if $N_R$ is itself redundant (i.e. if $N_{R'} \equiv 0$ for
some cut $R'$ that is closer to the output of $N$ than $R$). In this
case, it is still sufficient to use $T$ of one test because $N_R$ has
a CTS of one assignment (in terms of cut $R$).  Assume however, that
$N_R$ is non-redundant. In this case, there is no ``degenerate'' CTS
for $N_R$ and $T$ has to contain at least $|T_R|$ tests.  Assuming
that $T^*_R$ alone is far from being a CTS for $N_R$, a CTS $T$ for
$N$ will consist of many tests.

So a solution to the problem caused by redundancy of $N$ is as
follows.  One should require that for every cut $R$ where $N_R \equiv
0$ holds, set $T_R \cup T^*_R$ should be a CTS for $N_R$. The fact
that there always exists at least one cut $R$ where $N_R$ is
non-redundant eliminates degenerate single-test CTSs for $N$.

\section{Reusing SSAs}
\label{app:reuse}
Let \sas be applied to formula $G(V,W)$ to produce formula $H(V)$ and
its SSA. Let us explain the idea of SSA reusing by the following
example.  Let $P_0$ be the SSA generated by \sas in branch $w=0$ where
$w \in W$. Let us show how SSA $P_1$ for branch $w=1$ can be derived
from $P_0$.  Let $\Fi_0$ be the AC-mapping for $P_0$. Assume for the
sake of simplicity that
\begin{itemize}
\item  only one clause $B$ of $\Fi_0(P_0)$ contains literal $w$
\item  only assignment $\pnt{q} \in P_0$ is mapped by $\Fi_0$ to
  clause $B$.
\end{itemize}
Thus, the only reason why $P_0$ is not an SSA in branch $w=1$ is that
\pnt{q} is not mapped to any clause. (Recall that SSAs built by \sas
consist of assignments to $V$. So the construction of an SSA in branch
$w=1$ is different from $w=0$ only because some $V$-clauses of branch
$w=0$ are satisfied in branch $w=1$ and vice versa.) Let
\ti{BuildSSA}* denote the modification of procedure \ti{BuildSSA} (see
Figure~\ref{fig:bld_ssa}) aimed at re-using $P_0$ when building SSA
$P_1$.

Recall that \ti{BuildSSA} maintains sets $E$ and $Q$.  The former
consists of the assignments whose neighborhood has been already
explored and the latter stores the assignments whose neighborhood is
yet to be explored.  \ti{BuildSSA}* splits $Q$ into two sets: $Q'$ and
$Q''$. An assignment \pnt{p} is put in  $Q'$ if
\begin{itemize}
\item \pnt{p} is in $P_0$ and
 \item clause $\Fi_0(\pnt{p})$ is not satisfied by $w=1$
\end{itemize}
(In our case, every assignment of $P_0$ but the assignment \pnt{q}
above is put in set $Q'$.) On the other hand, every assignment whose
neighborhood is yet to be considered and that does not satisfy the two
conditions above is put in set $Q''$. The reason for this split is
that the assignments from $Q'$ are cheaper to process. Namely, if
$\pnt{p} \in Q'$, then instead of looking for a clause falsified by
\pnt{p}, \ti{BuildSSA}* uses clause $\Fi_0(\pnt{p})$. For that reason,
assignments of $Q'$ are the first to be considered by
\ti{BuildSSA}*. An assignment of $Q''$ is processed only if $Q'$ is
currently empty.

\ti{BuildSSA}* starts with the same center \sub{p}{init} that was used
when building $P_0$. If \sub{p}{init} is different from \pnt{q}, it is
put in $Q'$. Otherwise, it is put in $Q''$.  Let \pnt{p} be the
assignment picked by \ti{BuildSSA}* from $Q'$ or $Q''$. Let $C$ be the
clause to which \pnt{p} is mapped by $\Fi_1$. Let \pnt{p^*} be an
assignment of \Nbhd{p}{p}{C}. If \pnt{p^*} satisfies the two
conditions above, \ti{BuildSSA}* puts it in $Q'$. Otherwise, \pnt{p^*}
is added to $Q''$.

\end{document}